\documentclass[a4paper,11pt]{amsart}
\usepackage[left=2.5cm,right=2.5cm,top=2.5cm,bottom=2.5cm]{geometry}

\usepackage[utf8]{inputenc}
\usepackage{graphics}
\usepackage{graphicx}
\usepackage[version=3]{mhchem}
\usepackage{mathtools}
\usepackage{color}
\usepackage[english]{babel}
\usepackage{framed}
\usepackage{tikz}
\usepackage{enumerate}
\usepackage[sort&compress,comma,square,numbers]{natbib}
\usepackage{hyperref}
\usepackage{thmtools, thm-restate}
\usepackage{url}

\usetikzlibrary{arrows,decorations.pathmorphing,backgrounds,positioning,fit,matrix}

\newtheorem{Lemma}{Lemma}
\newtheorem{prop}{Proposition}
\newtheorem{criterion}{Criterion}

\theoremstyle{definition}
\newtheorem{definition}{Definition}
\newtheorem{example}{Example}

\renewcommand{\k}{{\kappa}}
\renewcommand{\t}{{\tau}}
\newcommand{\R}{\mathbb{R}}
\newcommand{\mmPT}{\mathcal{P}_T}

\DeclareMathOperator{\im}{Im}

\begin{document}

\title{Symbolic proof of bistability in reaction networks}
\author[A Torres, E Feliu]{Ang\'elica Torres$^1$, Elisenda Feliu$^{1,2}$}
\date{\today}

\footnotetext[1]{Department of Mathematical Sciences, University of Copenhagen, Universitetsparken 5, 2100 Copenhagen, Denmark.}
\footnotetext[2]{\emph{Corresponding author: } efeliu@math.ku.dk}

\begin{abstract}
Deciding whether and where a system of parametrized ordinary differential equations displays bistability, that is, has at least two asymptotically stable steady states for some choice of parameters, is a hard problem. For systems modeling biochemical reaction networks, we introduce a procedure to determine, exclusively via symbolic computations, the stability of the steady states for unspecified parameter values. In particular, our approach fully determines the stability type of all steady states of a broad class of networks.  To this end, we combine the Hurwitz criterion, reduction of the steady state equations to one univariate equation, and structural reductions of the reaction network. 
Using our method, we prove that bistability occurs in open regions in parameter space for many relevant motifs in cell signaling. 
\end{abstract}

\maketitle

\section{Introduction}
Bistability, that is, the existence of at least two stable steady states in a dynamical system, has been linked to switch-like behavior in biological networks and cellular decision making and it has been observed experimentally in a variety of systems \cite{Guidi:Bistability,Ninfape20,Ozbudak:BistabilityEcoli}. However, proving the existence of bistability in a parameter-dependent mathematical model is in general hard.

We focus on (bio)chemical reaction networks with associated kinetics, giving rise to systems of Ordinary Differential Equations (ODEs) that model the change in the concentration of the species of the network over time. These systems come equipped with unknown parameters and, ideally, one wishes to determine properties of the family of ODEs for varying parameter values. Here we are concerned with stability of the steady states, and focus on  the following three questions: (1) if the network admits only one steady state for all parameter choices, is it asymptotically stable?  (2) can parameter values be chosen such that the system is bistable? (3) does it hold that for any choice of parameters yielding at least three steady states, two of them are asymptotically stable?

As the parameters are regarded unknown, explicit expressions for the steady states are rarely available.  
Problem (1) has been shown to be tractable for certain classes of networks. For example,  the only positive steady state of complex balanced networks admits a Lyapunov function making it asymptotically stable \cite{Feinberg1972,hornjackson}. The use of Lyapunov functions and the theory of monotone systems has been employed more broadly to other classes of networks \cite{banaji-donnell-II,Ali-Angeli,angelisontag2}. Finally, algebraic criteria as the Hurwitz criterion or the study of P-matrices also provide asymptotic stability of steady states, often in combination with algebraic parametrizations of the steady state variety. In \cite{Clarke:graph} the Hurwitz criterion is analyzed using graphical methods.

Problem (2) is much harder and typically tackled by first deciding whether the network admits multiple steady states using one of the many available methods  \cite{joshi-shiu-III,crnttoolbox,MullerSigns,conradi-switch}, and then numerically computing the steady states and their stability for a suitable choice of parameter values.  Rigorous proofs of bistability require advanced analytical arguments such as bifurcation theory and geometric singular perturbation theory, as employed  in \cite{rendall-2site,rendall-feliu-wiuf} for futile cycles. The use of the Hurwitz criterion to prove bistability is anecdotal, as rarely explicit descriptions of the steady states can be found. Problem (3) has been addressed for small systems using case-by-case approaches, but no systematic strategy has been proposed.

We devise a flow chart to solve the problems (1)-(3) using computer-based proofs relying only on symbolic operations. This is achieved by combining three key ingredients. First, we apply the Hurwitz criterion \cite{Barnett:hurwitz-routh,Anagnost} on the characteristic polynomial of the Jacobian of the ODE system evaluated at a parametrization of the steady states. 
Second, we observe that when all but the last Hurwitz determinants are positive (meaning that instabilities only arise via an eigenvalue with positive real part), and further the solutions to the steady state equations are in one to one correspondence with the zeros of a univariate function, then the stability of the steady states is completely determined and question (3) can be answered  (Theorem~\ref{thm:bistability}). Third, if Theorem~\ref{thm:bistability} does not apply, whenever possible, we reduce the network to a smaller one for which Theorem~\ref{thm:bistability} applies. Afterwards, we use the results on the reduced network to infer stability properties of the steady states of the original network in an open parameter region. To this end, many reduction techniques have been proposed \cite{banaji-pantea-inheritance,joshi-shiu-II,inflows,craciun-feinberg,feliu:intermediates}, but often removal of reactions \cite{joshi-shiu-II} and of intermediates \cite{feliu:intermediates} suffice for (bio)chemical networks. 

Even though our approach demands  heavy symbolic computations, we illustrate how problem (3) can successfully be tackled  for small networks, and further, we prove the existence of bistability in open regions of the parameter space  for several relevant cell signaling motifs. 
In particular, for many networks  we provide the first proof of bistability  in an open region of the parameter space, with the exception of the double phosphorylation cycle, whose bistability was proven in \cite{rendall-2site}. 

This article is structured as follows: In Section~\ref{sec:framework}, we provide the mathematical framework  and state the two main results of this work. In Section~\ref{sec:procedure}, we present the procedure for detecting bistability and apply it to numerous examples arising from cell signalling. In Section~\ref{app:proofs} we prove the main results, and in Section~\ref{sec:comp_challenges} we address the computational challenges of our method. Finally, we conclude with a discussion of our contribution. We include an appendix with an  expanded explanation of the examples in Section~\ref{sec:procedure}, together with additional stability criteria that can be considered instead of the Hurwitz criterion. Computations are provided in an accompanying supplementary \texttt{Maple} file (available from the second author's webpage).

\section{Mathematical framework} \label{sec:framework}
We use the following notation: 
 $J_f$ is the Jacobian matrix of a function $f$. We denote by $V^\perp$  the orthogonal complement of a vector space $V$. For $A\in \R^{n\times n}$ and $I,J\subseteq \{1,\ldots ,n\}$,  we let $A_{I,J}$ be the submatrix of $A$ with rows (resp. columns) with indices in $I$ (resp. $J$). Finally, we denote by $p_A(\lambda)$  the characteristic polynomial of a matrix $A$.

\subsection{Reaction networks and steady states} We consider reaction networks over a set of \emph{species} $\mathcal{S}=\{X_1,\ldots ,X_n\}$  given by a collection of reactions
\begin{equation}\label{eq:reactions}
r_j \colon \sum\nolimits_{i=1}^n \alpha_{ij} X_i \longrightarrow \sum\nolimits_{i=1}^n \beta_{ij}X_i  \quad\mbox{ for } j=1,\ldots ,m
\end{equation}
	with $\alpha_{ij}\neq \beta_{ij}$ for at least one index $i$. 
Let $x_i$ denote the concentration of $X_i$. Given a differentiable  \emph{kinetics}   $v\colon \R^n_{\geq 0}\rightarrow \R^m_{\geq 0}$, the dynamics of the concentrations of the species in the network over time $t$ are modeled by means of a system of autonomous ODEs, 
\begin{equation}\label{eq:ode}
\tfrac{dx}{dt}=Nv(x), \qquad x=(x_1,\ldots ,x_n) \in \R^n_{\geq 0},
\end{equation}
where $N\in \mathbb{R}^{n\times m}$ is the \emph{stoichiometric matrix} with $j^{\rm th}$ column
 $(\beta_{1j}-\alpha_{1j},\ldots ,\beta_{nj}-\alpha_{nj})$. That is, the $j$-th column of $N$ encodes the net production of each species in the $j$-th reaction.  
We write  
\[ f(x)\coloneqq Nv(x).\]
With \textit{mass-action} kinetics,  we have $v_j(x)=\k_j x_1^{\alpha_{1j}}\cdots x_n^{\alpha_{nj}}$
where $\k_j>0$ is a \textit{rate constant}, shown often as a label of the reaction.
		
Under mild conditions, satisfied by common kinetics including mass-action, $\R^n_{\geq 0}$ and $\R^n_{>0}$ are forward invariant by Eq.~\eqref{eq:ode}, \cite{volpert}. Furthermore, any trajectory of Eq.~\eqref{eq:ode} 
is confined to a so-called \emph{stoichiometric compatibility class} $(x_0 + S)\cap \R^n_{\geq 0}$ with $x_0\in \R^n_{\geq 0}$,  where $S$ is the column span of $N$ and is called the \emph{stoichiometric subspace} \cite{feinberg-book}. 
We let $s=\dim(S)$ and $d=n-s$. The set $(x_0 + S)\cap \R^n_{\geq 0}$ is the solution set of the equations $Wx=T$ with 
$W\in \R^{d\times n}$ any matrix whose rows form a basis of $S^\perp$ and $T=Wx_0\in \R^d$. 
 These equations are called \emph{conservation laws} and the defined stoichiometric compatibility class is denoted by $\mathcal{P}_T$.  The vector $T$ is called a vector of \emph{total amounts}.

\smallskip
The steady states (or equilibria) of the network are the non-negative solutions to the system $Nv(x)=0$.  The positive steady states, that is, the solutions in $\R^n_{>0}$, define the  \emph{positive steady state variety} $V^+$.
The existence of $d$ linearly independent conservation laws implies that $d$ steady state equations are redundant. 
Let $W\in \R^{d\times n}$ be row reduced with $i_1,\ldots,i_d$ the indices of the first non-zero coordinate of each row. For $T\in \mathbb{R}^d$, we define the following function 
	\begin{equation}\label{F}
	F_T(x)_i=\begin{cases} 	f_i(x) & i\notin \{i_1,\ldots,i_d\}\\
		(Wx-T)_i & i\in \{i_1,\ldots,i_d\},
\end{cases}
	\end{equation}	
which arises after replacing redundant equations in the system $Nv(x)=0$ with $Wx-T=0$. Hence, the solutions to $F_T(x)=0$ are the steady states in the stoichiometric compatibility class $\mmPT$ \cite{FeliuPlos,wiuf-feliu}. 
A steady state $x^*$  is said to be \emph{non-degenerate} if $\ker(J_f(x^*))\cap S=\{0\}$, or equivalently,   if $\det(J_{F_T}(x^*))\neq 0$ \cite{wiuf-feliu}. Observe that $J_{F_T}(x^*)$ is independent of $T$.

\begin{example}\label{ex:HK}
Consider the following reaction network
		\begin{equation}\label{R.Ex}
			\ce{X1 ->[\k_1] X2 } \qquad \ce{X2 + X3 ->[\k_2] X1 + X4} \qquad \ce{X4 ->[\k_3] X3}.
		\end{equation}
		This is a simplified model of a two-component system, consisting of a histidine kinase HK and a response regulator RR \cite{FeliuPlos}. Both occur unphosphorylated ($X_1,X_3$) and phosphorylated ($X_2,X_4$). 
	With mass-action kinetics, the associated ODE system is
		\begin{align}
		\label{eq:x1x3} 	\tfrac{dx_1}{dt} &= -\k_1x_1+\k_2x_2x_3  &			\tfrac{dx_3}{dt} &= -\k_2x_2x_3 + \k_3x_4 \\
\label{eq:x2x4}			\tfrac{dx_2}{dt} &= \k_1x_1 -\k_2x_2x_3  & 			\tfrac{dx_4}{dt} &= \k_2x_2x_3 - \k_3x_4,
		\end{align} 
and we consider the stoichiometric compatibility class $\mmPT$ defined by $ x_1+x_2=T_1$ and $x_3+x_4=T_2$.
With this choice,  $F_T(x)$ is
\[\big( x_1+x_2-T_1,  \k_1x_1 -\k_2x_2x_3,  x_3+x_4-T_2,  \k_2x_2x_3 - \k_3x_4\big)^{tr},\]
where \emph{tr} indicates the transpose. 

In this example, the positive steady states are the positive solutions to the equations $ \k_1x_1 -\k_2x_2x_3=0,$ $ \k_2x_2x_3 - \k_3x_4=0$, 
obtained by setting the right-hand side of $\tfrac{dx_2}{dt},\tfrac{dx_4}{dt}$ in Eq.~\eqref{eq:x2x4}	 to zero.  
Solving these equations for $x_1$ and $x_3$ leads to the following parametrization of $V^+$:
		\begin{equation}\label{ss1}
		\phi(x_2,x_4)= \left(\tfrac{\k_3 x_4}{\k_1}, x_2 ,\tfrac{\k_3x_4}{\k_2x_2},x_4  \right),\qquad (x_2,x_4)\in \R^2_{>0}.
		\end{equation}
\end{example}	

In general, we refer to a \emph{positive parametrization} as any surjective map of the form
	\begin{equation}\label{eq:parametrization}
	\begin{split} 
 		\phi \colon \R^d_{>0} & \longrightarrow  V^+ \\
  		\xi & \longmapsto \phi(\xi). 
  	\end{split}
	\end{equation}
In practice, under mass-action kinetics, the entries of $\phi$ are  often rational functions in $\xi$. 
Positive parametrizations play a key role in what follows, and strategies to find one have been briefly reviewed in \cite{FeliuPlos}. Note that  there is no guarantee that a positive parametrization exists, or that can be found with known approaches. Nevertheless, networks arising from cell signaling with mass-action kinetics display features that facilitate the existence of parametrizations. The main such property is that monomials in $v(x)$ are often of at most degree two. This allows to perform linear elimination and write some variables in terms of positive functions of the others at steady state. If enough variables can be eliminated, then a positive parametrization can be found. This approach has been extensively studied, e.g. \cite{fwptm,Fel_elim,saez_elim,gunawardena-linear,TG-rational}.  The second scenario that favors finding parametrizations is when the positive steady state variety is described by binomial equations. In this case, $V^+$ is parametrized by monomials. Complex-balancing belongs to this scenario \cite{hornjackson,Feinberg1972,Craciun-Sturmfels,Dickenstein:2011p1112,muller}, as well as so-called networks with toric steady states \cite{PerezMillan}. Simple ways to find parametrizations for certain classes of networks (MESSI networks) have been also identified \cite{Dickenstein-MESSI}.   Finally, it is worth mentioning that often, brute force  with symbolic software (such as \texttt{Maple}) works: try to  solve symbolically the steady state equations for subsets of $s$ variables. If an output is produced, then decide whether/when the solution is positive.

\subsection{Multistationarity, bistability and network reduction}\label{sec:red}
We say that a network is \emph{multistationary} if it has at least two positive steady states in some $\mmPT$, that is, $F_T(x)=0$ has at least two positive solutions for some $T\in \R^d$. Similarly,  a \emph{monostationary} network has exactly one positive steady state in each $\mmPT$. 	

Under some conditions, if the sign of $\det(J_{F_T}(x^*))$ is $(-1)^s$ for all positive steady states  $x^*$, then the network is known to be monostationary; if the sign is $(-1)^{s+1}$ for some $x^*$, then it is multistationary \cite{FeliuPlos}. 
We a positive parametrization of $V^+$ is available, this result can be used to obtain inequalities in the rate constants that guarantee  or preclude multistationarity \cite{conradi-mincheva,FeliuPlos}.

Given an ODE system $\tfrac{dx}{dt}=f(x)$, a steady state $x^*$ is said to be \textit{stable} if for each  $\epsilon>0$, there exists $\delta>0$ such that solutions starting within distance $\delta$ of  $x^*$,  remain within distance $\epsilon$ of $x^*$.  If additionally  $\delta$ can be chosen such that solutions tend to $x^*$ as time increases, then $x^*$ is  \textit{asymptotically stable}. If $x^*$ is not stable, then it is  \textit{unstable}.

Stability of steady states can be often addressed by inspecting the eigenvalues of $J_f(x^*)$: if all eigenvalues of $J_f(x^*)$ have negative real part, then $x^*$ is said to be \emph{exponentially stable}, and  exponential stability implies asymptotic stability (\S2.7-2.8 in \cite{perko}). If at least one eigenvalue has positive real part, then $x^*$ is unstable. For further discussions on stability we refer to \cite{perko}. 

The stability   of a steady state is studied relatively to $\mmPT$. A network that admits at least two asymptotically stable positive steady states in some $\mmPT$ is called \emph{bistable}. 
Detecting multistationarity and bistability is challenging already for medium sized networks. To overcome computational difficulties one may employ several structural modifications of the network, among which we focus on removal of reactions (subnetworks) and removal of intermediates. 

Specifically, given two networks $\mathcal{G}$ and $\mathcal{G}'$, $\mathcal{G}'$ is a \textit{subnetwork} of $\mathcal{G}$ if it arises after removing reactions of $\mathcal{G}$  \cite{joshi-shiu-II}. Assuming mass-action kinetics and assuming that $\mathcal{G}$ and $\mathcal{G'}$ have the same stoichiometric subspace, if $\mathcal{G}'$ has $\ell_1$ exponentially stable and $\ell_2$ non-degenerate unstable steady states in some stoichiometric compatibility class  for some  rate constants $\k$, then $\mathcal{G}$ has at least $\ell_1$ exponentially stable and $\ell_2$ non-degenerate unstable steady states in some stoichiometric compatibility class  for some rate constants $\widetilde{\k}$. Here, $\widetilde{\k}$ agrees with $\k$ for the common reactions and is small enough for the reactions that  only are in $\mathcal{G}$.

An \emph{intermediate} $Y$ is any species that is both a product and a reactant in the network and that does not interact with any other species, that is, the only complex containing it is $Y$ itself. Removal of intermediates leads to a new network obtained after collapsing into one reaction all paths of reactions from and to non-intermediates and through intermediates \cite{feliu:intermediates}. For example, the species \ce{S0E} in 
\[  \ce{S_0 + E <=>S0E -> S1 + E} \]
 is an intermediate. Its removal yields the reaction \ce{S_0 + E -> S_1 + E}. 
 This is one of   the simplest forms of intermediates, but these can appear in complicated mechanisms linking several non-intermediate complexes.
 
 With mass-action kinetics, if $\mathcal{G}'$ is obtained from $\mathcal{G}$ through the removal of intermediate species, and satisfies a technical condition on the rate constants \cite{feliu:intermediates,amir_multi}, the following holds: If $\mathcal{G}'$ has $\ell_1$ exponentially stable and $\ell_2$ non-degenerate unstable steady states in some  stoichiometric compatibility class  for some  rate constants $\k$, then  $\mathcal{G}$ has at least $\ell_1$ exponentially stable and $\ell_2$ non-degenerate unstable steady states in some stoichiometric compatibility class for some   rate constants $\widetilde{\k}$. In \cite{feliu:intermediates,amir_multi}, an explicit description of the rate constants $\widetilde{\k}$ in terms of $\k$ can be found. 
The technical condition on the rate constants is satisfied by intermediates $Y$  appearing only in subnetworks of the form $y\ce{->} Y \ce{->} y'$ or $y\ce{<=>} Y \ce{->} y'$, where $y,y'$ are arbitrary complexes, see \cite[Prop. 5.3(i)]{amir_multi}. This is the main type of intermediates arising in cell signaling, and in particular in this work, and hence bistability and multistationarity can be lifted.

There are other network reduction techniques that yield to analogous results, see e.g. \cite{banaji-pantea-inheritance,inflows}. In what follows, we will only focus on removal of reactions and intermediates, as these modification typically suffice in networks arising from cell signaling.

\subsection{The Jacobian matrix of reaction networks}
In the context of reaction networks, we determine the stability of steady states based on the eigenvalues of the Jacobian of the restriction of the system in Eq.~\eqref{eq:ode} to a stoichiometric compatibility class. To this end, we consider the projection of  $J_f(x)$ onto the stoichiometric subspace $S$ by writing the ODE system in local coordinates of $S$. 
Let $R_0\in \mathbb{R}^{n\times s}$ be a matrix whose columns form a basis of $S$ and $L\in \R^{s\times m}$ such that $N=R_0 L$. Then the projection of  $J_f(x)$ onto   $S$ is  $L J_v(x) R_0$. 

\begin{prop}\label{prop:basic_jacobian}
Let $R_0\in \mathbb{R}^{n\times s}$ be a matrix whose columns form a basis of $S$, and $L\in \R^{s\times m}$ be such that $N=R_0 L$. For any $x\in \R^n$, let $Q_x:=L J_v(x) R_0$  and  denote the characteristic polynomial $p_{Q_{x}}(\lambda)$ by $q_{x}(\lambda)$.  
Then  the characteristic polynomials 
$p_{J_f(x)}$ and $q_{x}$ satisfy 
\[ p_{J_f(x)}(\lambda)=\lambda^{n-s}q_{x}(\lambda). \]
Further, the independent term of $q_{x}(\lambda)$ is  $(-1)^s\det(J_{F_T}(x))$, with $F_T$ as in Eq.~\eqref{F} for any  choice of $W$. 
\end{prop}

The proof of  Proposition~\ref{prop:basic_jacobian} can be found in Section~\ref{app:proofs} as part of Proposition~\ref{prop:basic_jacobian_II}.
According to Proposition~\ref{prop:basic_jacobian},  the $s$ eigenvalues of the matrix $Q_x$ are the eigenvalues of $J_f(x)$ once zero counted with multiplicity $d$ is disregarded. 
In order to study the (sign of the real part of the) spectrum of the matrices $Q_x$ when 
$x$ is a positive steady state, we use a positive parametrization.  
\begin{example}\label{ex:HK2}
For the network in  Eq.~\eqref{R.Ex}, we consider the following matrices
\[ R_0= \left[ \begin{array}{rr}
		-1 &0 \\[-2pt]
		1 & 0 \\[-2pt]
		0 & 1 \\[-2pt]
		0 & -1
	\end{array} \right], \quad  L= \left[ \begin{array}{ccc}
		1 & -1 & 0\\[-2pt]
		0 & -1 & 1
	\end{array} \right],\quad   J_v(x) = 
 			 \left[\begin{array}{cccc}
 			\k_1 & 0 & 0 & 0 \\[-2pt]
 			0 & \k_2x_3 & \k_2 x_2 & 0 \\[-2pt]
 			0 & 0 & 0 & \k_3 
 			\end{array}\right],
 	\]
and we are interested in the eigenvalues of the matrix
\begin{equation}\label{Qrest}
		 Q_{x}=LJ_v(x)R_0=\left[ \begin{array}{cc}
		-\k_1-\k_2x_3 & -\k_2x_2 \\[-2pt]
		-k_2 x_3 & -\k_2x_2-\k_3
	\end{array} \right]
	\end{equation}	
evaluated at a steady state  $x^*=\phi(x_2,x_4)$, c.f. Eq.~\eqref{ss1}.
Thus, by analyzing the eigenvalues of $Q_{\phi(x_2,x_4)}$  for \emph{all} values of   $\k$ and $x_2,x_4>0$, we study the stability of all positive steady states. In particular,  the characteristic polynomial $q_{x}$ of $Q_{x}$ in Eq.~\eqref{Qrest}  is 
\begin{equation*} 
q_{x}(\lambda) = \lambda^2+ (\k_2x_2+\k_2x_3+\k_1+\k_3)\lambda + \k_1\k_2x_2+\k_2\k_3x_3+\k_1\k_3.
\end{equation*}
\end{example}

\smallskip
We conclude this part with a key technical result (proven in Section~\ref{app:proofs}) on the determinant of $J_f(x^*)$ in the particular case where  system $F_T(x)=0$ is reduced to one univariate equation.  

\begin{restatable}[]{prop}{propreducepolynomial}
\label{prop:reduce_polynomial}%
Let $W\in \R^{d\times n}$ be a row reduced matrix whose rows form a basis of $S^\perp$, and let $T\in \R^d$ be fixed. Consider $F_T$ as in Eq.~\eqref{F}. 
Assume that there exist an open interval $\mathcal{E} \subseteq \mathbb{R}$, a differentiable function $\varphi\colon \mathcal{E}\rightarrow \mathbb{R}^n_{>0}$, and an index $j$ such that 
$ F_{T,\ell} ( \varphi(z) )=0$ for all $\ell\neq j$ and every $z\in\mathcal{E}$.  
Then, the set of positive solutions of the system $F_T(x)=0$ contains the solutions to
\begin{equation}\label{eq:onevariable}
	 F_{T,j} ( \varphi(z) )=0,\quad x_\ell = \varphi_\ell(z), \quad \ell=1,\ldots,n \mbox{ and } z\in \mathcal{E}.
\end{equation}
Furthermore, for a positive steady state $x^*=\varphi(z)$ such that $\varphi_i'(z)\neq 0$ for some $i$, it holds
\[\det (J_{F_T}(x^*))=\frac{(-1)^{i+j}}{\varphi_i '(z)}(F_{T,j}\circ \varphi)'(z) \det(J_{F_T}(x^*)_{J,I}),\]
where $I=\{1,\dots,n\}\setminus \{i\}$ and $J=\{1,\dots,n\}\setminus \{j\}$.
\end{restatable}

In practice, $\varphi_i(z)=z$, $\mathcal{E}\subseteq \R_{>0}$, and the solutions to $F_T(x)=0$ are in one to one correspondence with the solutions to Eq.~\eqref{eq:onevariable}. 
In this case, given the positive solutions $z_1<\dots < z_\ell$ of Eq.~\eqref{eq:onevariable}, 
 the sign of the derivative of $F_{T,j} ( \varphi(z) )$ evaluated at $z_1,\dots,z_\ell$ alternates if all the steady states are non-degenerate. 
If additionally the sign of $\frac{1}{\varphi_i'(z)} \det(J_{F_T}(\varphi(z))_{J,I})$ is independent of the choice of   $z$, then the sign of $\det (J_{F_T}(\varphi(z_\ell)))$ depends only on the sign of  the derivative of $(F_{T,j} \circ \varphi)$ at $z_\ell$ and hence alternates as well. We will exploit this fact below.

\subsection{Algebraic criteria for (bi)stability}
We present now the Hurwitz criterion \cite{Barnett:hurwitz-routh,Anagnost}, which is used to decide whether all the roots of a polynomial have negative real part, or whether there is at least one root with positive real part.

\begin{criterion}[Hurwitz]
	Let $p(x)=a_sx^s +\ldots  +a_0$ be a real polynomial with $a_s>0$ and $a_0\neq 0$. The Hurwitz matrix $H=(h_{ij})$ associated with $p$ has entries $h_{ij}=a_{s-2i+j}$ for $i,j=1,\ldots ,s$, by letting $a_k=0$ if   $k\notin \{0,\dots,s\}$:
	\[H=\left[\begin{array}{cccccc}
			a_{s-1} & a_s & 0 & 0 &  \cdots & 0 \\[-2pt]
			a_{s-3} & a_{s-2} & a_{s-1} & a_s & \cdots & 0 \\[-2pt]
			\vdots & \vdots & \vdots &  \vdots & \vdots & \vdots \\[-2pt]
			0 & 0 & 0 & a_{6-s} & \cdots & a_2 \\[-2pt]
			0 & 0 & 0 & 0 &\cdots & a_0
	\end{array}\right]\in \R^{s\times s}.\]
The $i^{\rm th}$ Hurwitz determinant is defined to be $H_i=\det(H_{I,I})$, where $I=\{1,\ldots ,i\}$. 
Then, all roots of $p$ have negative real part if and only if $H_i>0$ for all $i=1,\ldots ,s$. 
If $H_i<0$ for some $i$, then at least one root of $p$ has positive real part. 
\end{criterion}

Importantly, $H_s=a_0H_{s-1}$.
Pairs of imaginary roots (leading to Hopf bifurcations) arise when at least one of the $H_i$  vanish \cite{yang-hopf} (see \cite{Shiu-Hopf} in the context of reaction networks).

For the polynomial $q_{x}(\lambda)$ or $q_{\phi(\xi)}(\lambda)$,  the Hurwitz determinants are typically rational functions in $x$ or $\xi$. 

\begin{example}
In Example~\ref{ex:HK2},
the Hurwitz determinants of the characteristic polynomial $q_x$ are
\[H_1  =  \k_2(x_2+x_3)+\k_1+\k_3,\quad\textrm{ and }\quad \ H_2 =( (\k_1x_2+\k_3x_3)\k_2+\k_1\k_3)H_1.\]
Both determinants are polynomials in $\k,x$ with positive coefficients and hence positive for all $\k\in\R^3$ and $x\in \R^4_{>0}$. 
By the Hurwitz criterion, any  positive steady state is exponentially stable. 
This network has exactly one steady state in each stoichiometric compatibility class  \cite{FeliuPlos}, and  we now additionally conclude that \emph{the only steady state is exponentially stable}.
\end{example}
 
Often for small networks, all but the last Hurwitz determinants are positive. 
Then,  the stability of a steady state $x^*$ is fully determined by the sign of $H_s$, which 
agrees with the sign of the independent term of $q_{x^*}(\lambda)$, which in turn is  $(-1)^s\det(J_{F_T}(x^*))$ by Proposition \ref{prop:basic_jacobian}. Together with Proposition~\ref{prop:reduce_polynomial}   we obtain the following theorem, proved in Section~\ref{app:proofs}.

\begin{restatable}[]{thm}{thmbistability}
\label{thm:bistability}%
Let $T,\mathcal{E},\varphi, i, j,I,J$ be as in Proposition~\ref{prop:reduce_polynomial}, and $q_{x^*}$ be the characteristic polynomial of the matrix $Q_{x^*}$ as in Proposition~\ref{prop:basic_jacobian}. Assume that 
\begin{itemize} 
\item the sign of $\frac{1}{\varphi_i'(z)}\det(J_{F_T}(\varphi(z))_{J,I})$ is independent of $z\in \mathcal{E}$ and is nonzero, and
\item the first $s-1$ Hurwitz determinants of $q_{x^*}$ are positive for all positive steady states $x^*$.
\end{itemize}				
If $z_1<\dots < z_\ell$ are the positive solutions to  Eq.~\eqref{eq:onevariable} and all are simple, then either $\varphi(z_1),\varphi(z_3),\dots$ are exponentially stable and $\varphi(z_2),\varphi(z_4),\dots$ are unstable, or the other way around. 
Specifically, $\varphi(z_1)$ is exponentially stable if and only if 
\begin{equation}\label{eq:importantsign}
\frac{(-1)^{s+i+j}}{\varphi_i'(z_1)}(F_{T,j}\circ \varphi)'(z_1) \det(J_{F_T}(\varphi(z_1))_{J,I})>0.
\end{equation}
\end{restatable}

\begin{figure}[t!]
	\centering
	\includegraphics[scale=1]{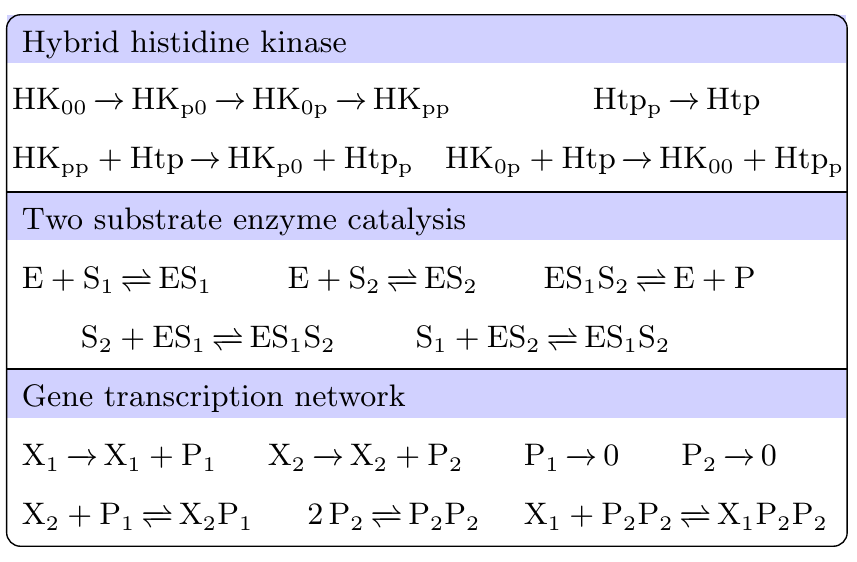} 
	\caption{{ \footnotesize Three networks where stability of steady states is fully determined.} }	
	\label{Figure 1}
\end{figure}

In practice, 
$ F_{T,j}(\varphi(z))=\frac{a(z)}{b(z)}$ is a rational function in $z$ with  $b(z)>0$ in $\mathcal{E}$. Then the zeros of $F_{T,j}(\varphi(z))$ are the roots of $a(z)$, and the signs of  $( F_{T,j} \circ \varphi)'(z^*) $ and $a'(z^*)$ agree for all $z^*\in \mathcal{E}$ such that $a(z^*)=0$ (Lemma~\ref{lemma:rationalfunction} in Section \ref{app:proofs}).

\begin{example}\label{ex:hybrid}
We illustrate Theorem~\ref{thm:bistability} with a \textbf{hybrid histidine kinase}  network with mass-action kinetics~\cite{feliu:unlimited}, see Fig.~\ref{Figure 1}. 
We rename  the species as follows: $X_1$=\ce{HK00}, $X_2$=\ce{HK_{p0}}, $X_3$=\ce{HK_{0p}}, $X_4$=\ce{HK_{pp}}, $X_5$=\ce{Htp} and $X_6$=\ce{Htp_p}.  The associated ODE system is 
    		\begin{align*}
\tfrac{dx_1}{dt} &=  -\k_1 x_1  + \k_{4} x_3x_5  & \tfrac{dx_4}{dt}  &= \k_{3} x_3     -\k_{5} x_4 x_5 \\  
\tfrac{dx_2}{dt} &= \k_1 x_1  - \k_{2} x_2 +\k_5 x_4 x_5  & \tfrac{dx_5}{dt} &= -\k_{4} x_3x_5    -\k_{5} x_4x_5  +\k_{6} x_6 \\
\tfrac{dx_3}{dt} &=  -\k_{3} x_3    +\k_{2} x_2 -\k_4 x_3 x_5  & \tfrac{dx_6}{dt} &= \k_{4} x_3x_5     -\k_{6} x_6 +\k_5 x_4 x_5.  
			\end{align*}
The conservation laws of the system are $T_1= x_1+x_2+x_3+x_4$, $T_2= x_5+x_6$. 
Hence, by Eq.~ \eqref{F}, 
\[ F_T(x)= \left[ \begin{array}{c} x_{1}+x_{2}+x_{3}+x_{4}-T_{1}
\\ \k_{5}x_{4}x_{5}+\k_{1}x_{1}-\k_{2}x_{2
}\\ -\k_{4}x_{3}x_{5}+\k_{2}x_{2}-\k_{3}x_{3}\\ -\k_{5}x_{4}x_{5}+\k_{3}x_{3}
\\ x_{5}+x_{6}-T_{2}\\ \k_{4}x_{3}x_{5}+\k_{5}x_{4}x_{5}-\k_{6}x_{6}\end {array}
 \right].
\]
Here $s=4$. 
The existence of three positive steady states for this network was established in~\cite{feliu:unlimited}.
We compute $q_x$ and the Hurwitz determinants in {\tt Maple 2019} (see accompanying \texttt{Maple} file) and obtain that all but the last are polynomials in $x$ and $\k$ with positive coefficients; hence they are positive when evaluated at a positive steady state.

We proceed to decide whether Theorem~\ref{thm:bistability} applies.  In~\cite{feliu:unlimited},
 it was shown that the assumptions of Proposition~\ref{prop:reduce_polynomial} hold with $i=j=5$, with $z=x_5$, $F_{T,5}$ corresponding to the conservation law with $T_2$, and $\mathcal{E}=\R_{>0}$.  That is, the solutions to the four steady state equations together with the conservation law associated with $T_1$ can be parametrized by a function $\varphi$ that only depends on $z=x_5$. The denominator of $(F_{T,5}\circ\varphi)(z)$  is positive and its numerator is a polynomial of degree $3$ in $z$, which can have $1,2$ or $3$ positive roots, depending on the choice of parameters. Additionally, $\det(J_{F_T}(\varphi(x_5))_{J,I})$ is a rational function with all coefficients positive.  
Thus, we are in the situation of Theorem~\ref{thm:bistability}. The independent term of the numerator of  $(F_{T,5}\circ\varphi)(z)$ is negative, its first root has positive derivative.
Furthermore, the sign of $\frac{(-1)^{s+i+j}}{\varphi'_5(z)}=(-1)^{4+5+5}$ in Eq.~\eqref{eq:importantsign} is $+1$. Hence, if $z_1$ is the smallest positive root of $(F_{T,5}\circ\varphi)(z)$, the sign of 
$$\frac{(-1)^{s+i+j}}{\varphi_i'(z_1)}(F_{T,j}\circ \varphi)'(z_1) \det(J_{F_T}(\varphi(z_1))_{J,I})$$ 
is positive as well. 
 We conclude, using Theorem \ref{thm:bistability}, that whenever the network has three positive steady states coming from the roots $z_1<z_2<z_3$ of $(F_{T,5}\circ\varphi)(z)$, then the steady states $\varphi(z_1)$ and $\varphi(z_3)$ are  exponentially  stable and the steady state $\varphi(z_2)$ is unstable.  We have shown that \emph{this network displays bistability whenever there are three positive steady states}.
 \end{example}

\medskip
Theorem~\ref{thm:bistability} implies the existence of  bistability if the network admits at least three positive steady states and  \eqref{eq:importantsign} holds. Indeed, as the stability of the positive steady states alternates and that the first one is exponentially stable (when ordered as in Theorem~\ref{thm:bistability}),  there exist at least two exponentially stable steady states.
If \eqref{eq:importantsign} does not hold, then bistability will only follow if the network admits four positive steady states.
 
If $F_{T,j}(\varphi(z))$ is a rational function with positive denominator, we verify \eqref{eq:importantsign} by   finding the sign of $\frac{1}{\varphi_i'(z)}\det(J_{F_T}(\varphi(z))_{J,I})$ (which is independent of $z$), and inspecting the sign of the independent term of the numerator of $F_{T,j}(\varphi(z))$ to deduce the sign of 
$(F_{T,j}\circ \varphi)'(z_1)$.

In order to establish that the network admits three positive steady states, several strategies can be employed. 
First, one can directly attempt to find values of the parameters for which the numerator of $F_{T,j}(\varphi(z))$ has three positive roots. By Descartes' rule of signs and the degree of the polynomial, upper bounds on the possible number of positive roots can easily be found.

Alternatively, numerous existing methods to determine multistationarity \cite{crnttoolbox,MullerSigns,conradi-switch,PerezMillan,FeliuPlos,control} can be employed, see \cite{joshi-shiu-III} for a review. Most of these methods give a choice of (rational) parameters for which the network has \emph{at least two} positive steady states, but we require \emph{three}. 
An exception is the method in \cite{dickenstein:regions}, which via the study of  the Newton polytope of the system  of interest,  might directly certify the existence of at least three positive solutions. 

For the purpose of finding three positive  steady states, perhaps the simplest approach is to consider a choice of parameters for which multistationarity exists, as for example returned by the CRNT toolbox \cite{crnttoolbox}, and find the solutions to the corresponding system $F_T(x)=0$ (or equivalently $F_{T,j}(\varphi(z))=0$) using available software. This will typically provide the three desired solutions.  Note that this last step is not necessarily symbolic. To certify  that chosen parameters give rise to three positive steady states using symbolical methods, we use Sturm sequences on the numerator of  $F_{T,j}(\varphi(z))$  to count the number of positive roots \cite{basu:book}. 

In this work, we have mainly applied the method  in \cite{FeliuPlos} to determine multistationarity. The approach works for mass-action kinetics and relies on the  determinant of $J_{F_T}(\phi(\xi))$, for a positive parametrization $\phi$. Assume the network is conservative and has no  steady state at the boundary of  any  $\mathcal{P}_T$. Then, if $\k$  and $\xi$ are such that  $(-1)^s\det( J_{F_T}(\phi(\xi)))<0$, then the network displays multistationarity in the stoichiometric compatibility class containing $\phi(\xi)$.

\begin{figure}[b!]
\centering
\includegraphics[scale=0.87]{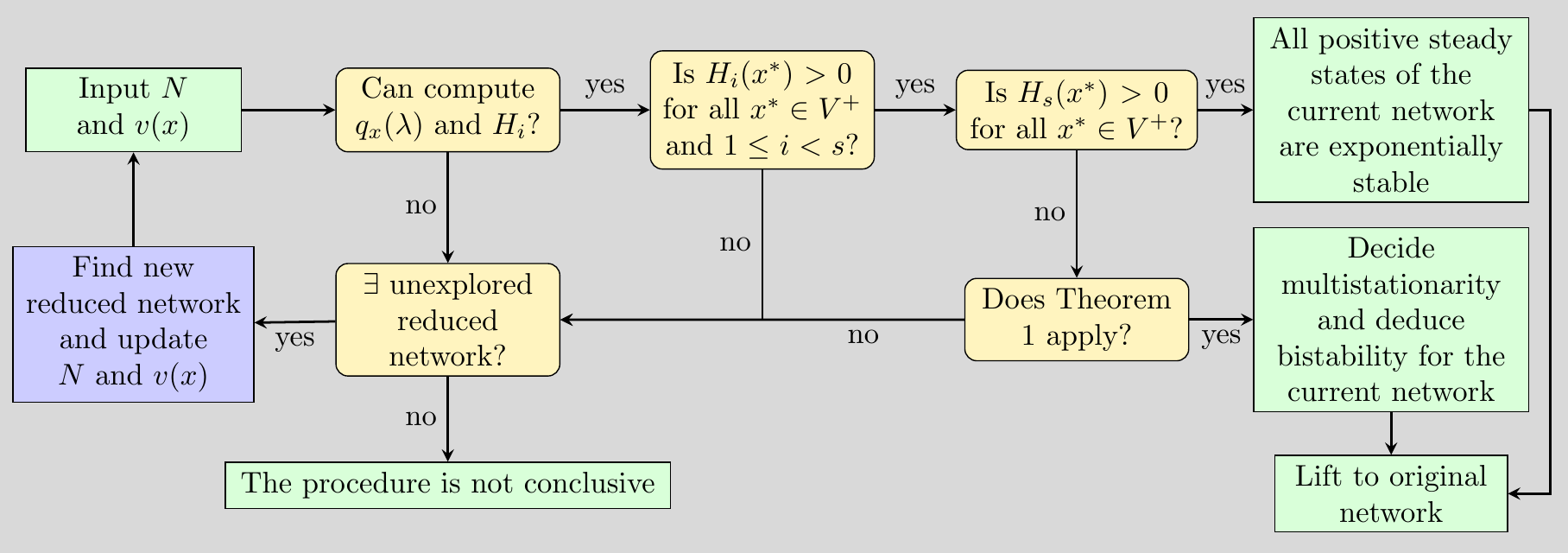}
\caption{{\footnotesize Flow chart of the proposed approach to study the stability of steady states.  $H_i$ denotes the $i$-th Hurwitz determinant of $q_{x}$. 
To decide whether $H_i(x^*)>0$ for all $x^*\in V^*$ and some $i$,   first check whether $H_i(x)>0$ for all $x\in \R^n_{>0}$. If not, find a positive parametrization $\phi$ (if possible, see text), and decide whether $H_i(\phi(\xi))>0$ for all $\xi\in \R^d_{>0}$. If a parametrization cannot be found, or the statement $H_i(x^*)>0$ for all $x^*\in V^*$  cannot be verified, we follow the ``no'' arrow out of the box in the decision diagram. The lift to the original network means that the conclusions on the currently analyzed network also hold for the original network in an open parameter region.  
}}\label{fig:method}
\end{figure}
\section{Symbolic determination of stability}\label{sec:procedure}

We now combine the ingredients introduced in the previous section into a strategy to determine the stability of positive steady states  and, importantly, detect bistability, using (mainly) the Hurwitz criterion and Theorem~\ref{thm:bistability}. 
  Given a reaction network with kinetics $v(x)$
the steps taken are depicted in Fig.~\ref{fig:method}. Specifically, we find the characteristic polynomial $q_x$ and the Hurwitz determinants. If all determinants are positive, then all positive steady states are exponentially stable. If only the last Hurwitz determinant can be negative, then we attempt to apply  Theorem~\ref{thm:bistability}. We only find a parametrization $\phi$ as in Eq.~\eqref{eq:parametrization} when the sign of $H_i$ is not already determined for arbitrary   $x\in \R^n_{>0}$, that is, only when it is necessary to evaluate at a positive steady state to decide the sign of $H_i$.
 
If some of the steps fail, then we consider reduced networks by  removing either reactions that do not change the stoichiometric subspace, or intermediates. If stability is determined for a reduced network, then we conclude that the original network has at least the same number of positive steady states and stability properties as the reduced network in an open parameter region. In particular, if the reduced network has bistability, then so does the original network. For an expansion on how the open parameter regions arise when lifting rate constants and the properties of steady states, we refer the reader to~\cite{feliu:intermediates,joshi-shiu-II}.

If stability is not determined for the chosen reduced network, then we consider another one, until all possibilities have been explored. If a reduced network does not display multistationarity, then further reductions on this network lead to networks that neither are multistationary. Being a reduced network of a reduced network establishes a partial order in the set of all reduced networks obtained by removal of intermediates or reactions. Therefore, a suitable strategy is to start investigating the largest networks. If one such network is not multistationary, all reduced networks smaller than this one can be disregarded.

\smallskip
We now use  this approach on the remaining networks in Fig.~\ref{Figure 1}.
We consider a \textbf{two substrate enzyme catalysis} mechanism, comprising an enzyme $E$ that binds two substrates, $S_1$ and $S_2$, and catalyzes the reversible conversion to $P$.  Taken with mass-action kinetics this network has one  positive steady state in each stoichiometric compatibility class for any choice of rate constants $\k$~\cite{FeliuPlos}. 
All but the last of the four Hurwitz determinants  are positive for $x\in \R^6_{>0}$. We find a positive parametrization $\phi$ by  solving the steady state equations in the concentrations of \ce{ES1}, \ce{ES2}, \ce{ES1S2} and \ce{P}  using the procedure  in  \cite{Fel_elim}, see Appendix~\ref{app:examples}. After evaluating at $\phi$, $H_4$ becomes a rational function with only positive coefficients. Hence, all Hurwitz determinants are positive   at a positive steady state, showing that \emph{the only positive steady state in each stoichiometric compatibility class is exponentially stable.}

  \begin{figure}[t!]
	\centering
	\includegraphics[scale=1]{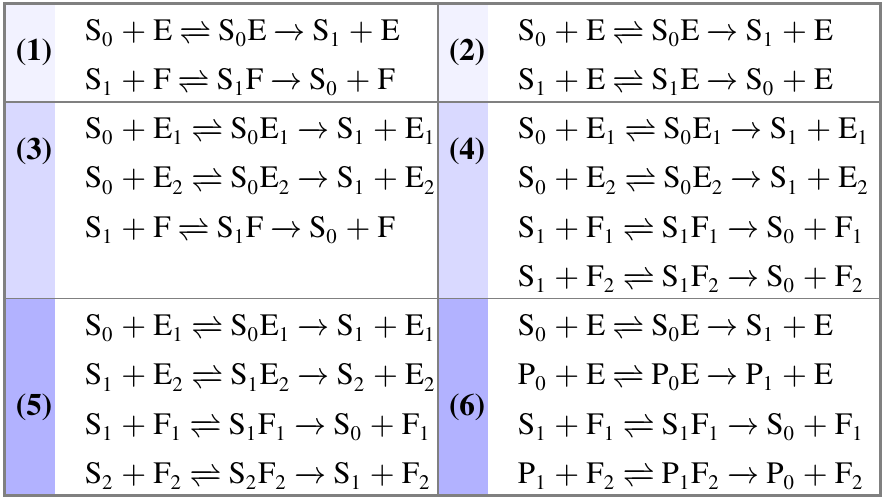} 
	\caption{{\footnotesize Monostationary networks. In all networks, the symbols \ce{E}, \ce{F}, \ce{S}, \ce{P} refer to kinases, phosphatases, and substrates respectively. Taken with mass-action kinetics, all networks admit exactly one positive steady state in each stoichiometric compatibility class, which further is exponentially stable. Networks (1)-(4) model the phosphorylation of one substrate via different mechanisms. Network (5) models a substrate with two phosphorylation sites, while network (6) models the phosphorylation of two different substrates.} }		
	\label{Figure 2}
\end{figure}
\smallskip
Next, we consider the \textbf{gene transcription network} in Fig.~\ref{Figure 1} with mass-action kinetics.
For any choice of   rate constants there exist  at least two positive steady states  in some  stoichiometric compatibility class \cite{FeliuPlos}. 
The computation of the Hurwitz determinants for arbitrary $x\in \R^7_{>0}$ gives that only $H_1,H_2$ are positive, but after evaluating at a positive parametrization, all but the last Hurwitz determinants are positive. We proceed to verify the assumptions of Theorem~\ref{thm:bistability}, see Appendix~\ref{app:examples}. We obtain that the maximum number of positive steady states in any stoichiometric compatibility class is three, and that, \emph{whenever the network has one positive steady state, then it is exponentially stable, and if it has three positive steady states, then two of them are exponentially stable and one is unstable. }
Hence bistability arises in this network.

\subsection*{Bistability in cell signaling}
After having illustrated our approach with selected examples, we now investigate relevant cell signaling motifs.
We follow the flow chart in Fig.~\ref{fig:method} and perform all computations  in \texttt{Maple 2019} (see accompanying \texttt{Maple}  file).
To reduce computational cost, we compute first the Hurwitz determinants of a generic degree $n$ polynomial (with unspecified coefficients), and then evaluate the determinant at the  coefficients of $q_x$. 
We disregarded the Routh table from the package {\tt DynamicSystems} as the computational cost was higher.
	
All networks in Fig.~\ref{Figure 2} are known to be \textbf{monostationary} under mass-action kinetics  \cite{Feliu:royal}. All Hurwitz determinants of $q_x$ are positive for positive  $x$, without the need of a positive parametrization, see Appendix~\ref{app:examples}. Hence for any choice of rate constants, each network in Fig.~\ref{Figure 2} \emph{has exactly one positive steady state in each stoichiometric compatibility class, which further is exponentially stable}.

\begin{figure*}[t!]
	\centering
	\includegraphics[scale=0.93]{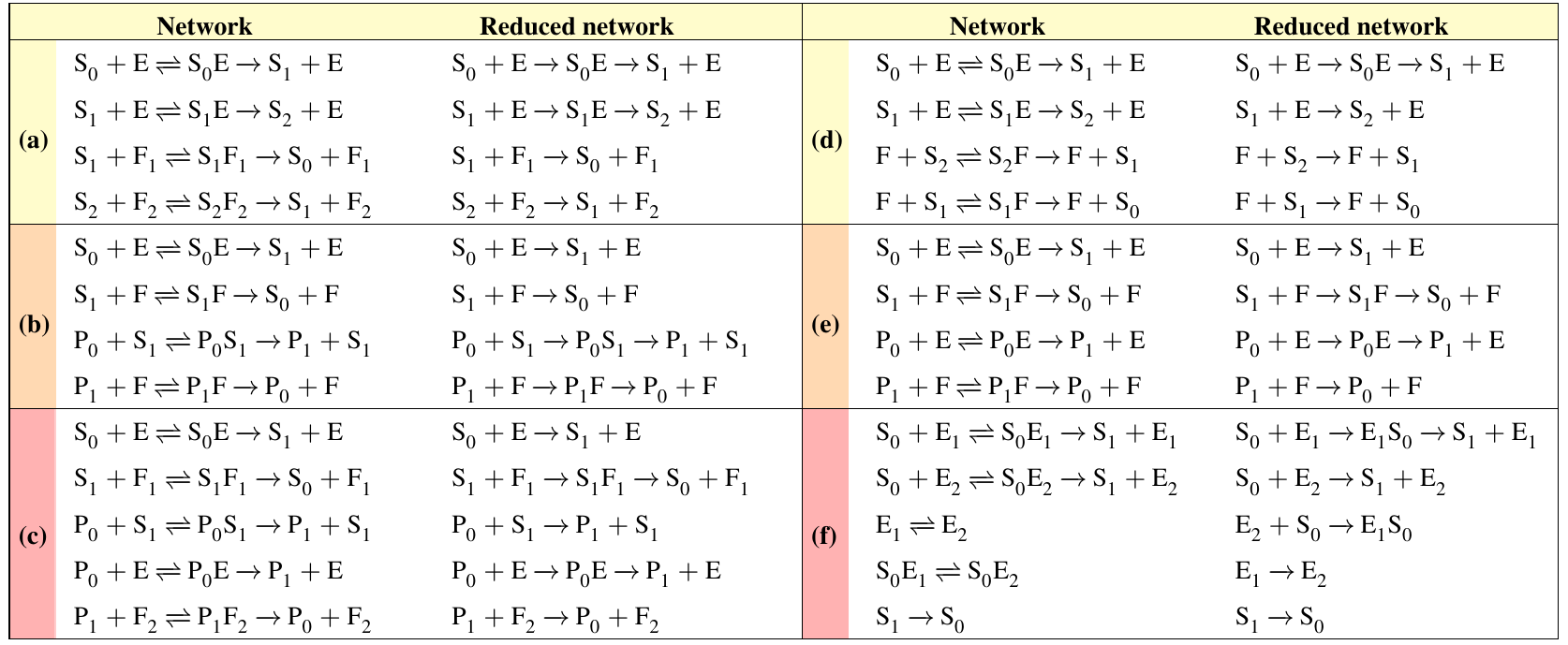} 
	\caption{{\footnotesize Multistationary networks and reductions to assert bistability. (a) \ce{E1} and \ce{E2} are two conformations of a kinase catalyzing the phosphorylation of  \ce{S0} \cite{feng:allosteric}. The reduced network is obtained by removing the intermediate \ce{E2 S0} and all reverse reactions. (b)  Cascade of two one-site modification cycles with the same phosphatase \ce{F}. The reduced network is obtained by removing the intermediates \ce{S0 E} and \ce{S1 F} and all reverse reactions. (c) Cascade of one-site modification cycles where the same kinase \ce{E} acts in both layers. The reduced network is obtained by removing the reverse reactions and intermediates \ce{S0 E}, \ce{P0 S1} and \ce{P1F2}.
(d) Distribute and sequential phosphorylation of a substrate. The reduced network is obtained by removing the intermediates \ce{S1 E}, \ce{S2 F} and \ce{S1 F} and all reverse reactions. (e) Phosphorylation of two substrates by the same kinase and phosphatase. The reduced network is obtained by removing the intermediates \ce{S0 E} and \ce{P1 F} and all reverse reactions. (f) Double phosphorylation of a substrate by the same kinase and two different phosphatases. The reduced network is obtained by removing all reverse reactions and the intermediates \ce{S1 F1} and \ce{S2 F2}. }}		
	\label{Figure 3}
\end{figure*}

In the examples so far, we have not employed network reduction techniques, because all steps of the method could be carried through and  the stability of a steady state depended only on the sign of the determinant of the Jacobian. This scenario is quite restrictive, as it implies that instabilities arise from a unique eigenvalue with positive real part.

The networks in Fig.~\ref{Figure 3} are all known to be \textbf{multistationary} with mass-action kinetics  \cite{Feliu:royal,feng:allosteric}.
Our method fails on the original networks: for network (c), the computation of the Hurwitz determinants was not possible in a  regular PC due to lack of memory, and for the rest of the networks  $H_i>0$ for $i\neq s$ does not hold.
However,  Theorem~\ref{thm:bistability} applies to the reduced networks  in Fig.~\ref{Figure 3}. In particular, all reduced networks in Fig.~\ref{Figure 3} display bistability whenever they have three positive steady states. In this case, two of the steady states are exponentially stable and the third is unstable. Hence, after lifting stability to the original network using the results stated in Subsection~\ref{sec:red},  \emph{for all networks in Fig.~\ref{Figure 3}, there is an open parameter region where the network has two exponentially stable positive steady states}.
We report only on a maximal reduced network that allows us to assert bistability for the original network. This means that, in practice, we might have checked several other reduced networks where Theorem~\ref{thm:bistability} did not apply.

See Appendix~\ref{app:examples} and the accompanying \texttt{Maple}  file for details.

\section{Proofs of the main results}\label{app:proofs}
We prove here the main results of our work. For completeness, we include the notation and the statements here, sometimes in an expanded form.

\subsection{Proof of Proposition~\ref{prop:basic_jacobian}} 
Consider a matrix $R_0\in \mathbb{R}^{n\times s}$ whose columns form a basis of $S$. This basis gives the system of coordinates in $S$. Therefore, given coordinates $z=(z_1,\ldots ,z_s)$ in $S$, the vector $R_0z$ is the vector of coordinates in the canonical basis of $\mathbb{R}^n$. Conversely, selecting a matrix $R_1\in \mathbb{R}^{s\times n}$ such that $R_1R_0=I_{s\times s}$, we can write a vector $x\in S$ given in the canonical basis of $\mathbb{R}^n$, as a vector in local coordinates, by performing the product $R_1x$.

Using these matrices, the ODE system restricted to  $(x^*+S)\cap \R^n_{\geq 0}$ in local $S$ coordinates is
		\[\dot{z}=R_1f(R_0z+x^*)\]
after translating the steady state $x^*$ to the origin.
The Jacobian matrix associated with this system at $0$ is $R_1 J_f(x^*) R_0$. The following proposition shows some basic properties of $R_1 J_f(x^*) R_0$, and includes the results stated in Propostion~\ref{prop:basic_jacobian}. These properties follow from basic linear algebra and, for example, statements (i) and (ii) in Proposition~\ref{prop:basic_jacobian_II} appear in \cite[Appendix A]{Banaji-Pantea}. We include a proof for completeness.  

\begin{prop}\label{prop:basic_jacobian_II}
Consider a reaction network with rate function $f(x)=Nv(x)$. Let $R_0\in\mathbb{R}^{n\times s}$ and $R_1\in\mathbb{R}^{s\times n}$ be matrices such that the columns of $R_0$ form a basis of the stoichiometric subspace $S$, and $R_1R_0=I_{s\times s}$. The following statements hold:

\begin{enumerate}[(i)]
\item $R_1J_f(x^*)R_0 = L J_v(x^*) R_0$, where $L\in \R^{s\times m}$ is the matrix such that $N=R_0 L$. In particular, $R_1J_f(x^*)R_0$ does not depend on the choice of $R_1$.

\item If $R_0,R_0'\in \R^{n\times s}$ are two matrices with column span $S$, and $L,L'$ are as in (i) for $R_0,R_0'$ respectively, then the matrices $L' J_v(x^*) R'_0$ and $L J_v(x^*) R_0$ are similar.

\item For a positive steady state $x^*$, the characteristic polynomials $p_{J_f}(\lambda)$ and $p_{LJ_v(x^*)R_0}(\lambda)$ satisfy $p_{J_f}(\lambda)=\lambda^{n-s}p_{LJ_v(x^*)R_0}(\lambda)$ for any choice of $R_0$. 

\item The independent term of $p_{LJ_v(x^*)R_0}(\lambda)$ (or the coefficient of degree $n-s$ of $p_{J_f}(\lambda)$) equals $(-1)^s\det(J_{F_T}(x^*))$, with $F_T$ as in Eq.~\ref{F}, for any choice of row-reduced matrix of conservation laws $W$. 
\end{enumerate}
\end{prop}

\begin{proof}
\textit{(i)} Since the columns of $N$ belong to $S$, we can uniquely write $N=R_0L$ with $L\in \R^{s\times m}$. 
Given that $J_f(x)=NJ_v(x)$, we have 
 \[ R_1J_f(x^*)R_0 =   R_1NJ_v(x^*) R_0= R_1R_0L J_v(x^*) R_0=L J_v(x^*) R_0. \] 				
\textit{(ii)} Let $M\in \R^{s\times s}$ be the matrix of change of basis from $R_0$ to $R_0'$ such that $R_0M=R_0'$.
From $N=R_0'L'=R_0L$, it follows that $R_0ML'=R_0L$ and thus $L'=M^{-1}L$. This gives
\begin{equation*}
L' J_v(x^*) R'_0 = M^{-1} LJ_v(x^*)R_0M,
\end{equation*}
which implies that $L'J_v(x^*) {R'}_0$ and $LJ_v(x^*)R_0$ are similar.
\newline 
\textit{(iii)} Extend the matrix $R_0$ to a square matrix $R\in \R^{n\times n}$ by adding columns such that $R$ has full rank $n$.
Then the eigenvalues of the matrices $Q=R^{-1}J_f(x^*)R$ and $J_f(x^*)$ coincide.
We choose  $R_1$ as the first $s$ rows of $R^{-1}=\begin{pmatrix}
R_1 \\ R_1'
\end{pmatrix} $. Then $R_1R_0=I_{s\times s}$ and $R_1'R_0=0$. Since  $\im(J_f(x^*))\subseteq S$, 
 the column span of $J_f(x^*)$ is contained in the column span of $R_0$, and thus $R_1'J_f(x^*) =0$.
Then, the matrix $Q$ has the form
\[ Q= \begin{pmatrix}
R_1 \\ R_1'
\end{pmatrix} J_f (x^*)  \begin{pmatrix}
R_0 &  R_0'
\end{pmatrix}  =  \begin{pmatrix}
R_1 J_f(x^*) \\ 0
\end{pmatrix}   \begin{pmatrix}
R_0 &  R_0'
\end{pmatrix}   = \left(  \begin{array}{cc}
			R_1 J_f(x^*)R_0  & R_1 J_f(x^*) R_0' \\
			0 & 0 
		\end{array}\right).
		\]
Clearly, the characteristic polynomial $p_Q(\lambda)$ is equal to $$\lambda^{n-s}p_{R_1 J_f(x^*)R_0}(\lambda).$$ Using {\small $R_1 J_f(x^*)R_0= L J_v(x^*)R_0$}, this concludes the proof of (iii). 
\newline
\textit{(iv)} The statement was proven in  \cite{wiuf-feliu}, Proposition 5.3.
\end{proof}

\subsection{Proof of Proposition \ref{prop:reduce_polynomial}}

We now turn into the proof of Proposition~\ref{prop:reduce_polynomial}. 
		In order to prove the identity in the statement,  we make first three observations that rely on the definition of $\varphi$ and on the chain rule for multivariate functions:
		\begin{enumerate}
			\item[(1)] By hypothesis, $\varphi(z)=\left(\varphi_1(z),\ldots ,\varphi_n(z)\right)$. Therefore, $\varphi'(z)=\left(\varphi_1'(z),\ldots ,\varphi_n'(z)\right)$.
			\item[(2)] Since $F_{T,\ell}(\varphi(z))=0$ for all $\ell\neq j$, we have that $(F_T\circ\varphi)(z)$ is a vector with zeros in every entry except for the $j$-th entry, which is equal to $(F_{T,j} \circ\varphi)(z) $. This implies that $(F_T\circ\varphi)'(z)$ is also a vector with zero in every entry except in the $j$-th, that is equal to $(F_{T,j} \circ\varphi)'(z)$. 
			\item[(3)] By the chain rule $(F_T\circ\varphi)'(z)=J_{F_T}(\varphi(z)) \varphi '(z)$. 
			
		\end{enumerate}
		From observations (2) and (3) we have 
\begin{equation}\label{eq:id}
J_{F_T}(\varphi(z)) \varphi '(z)=\left( 0,\ldots ,0,(F_{T,j}\circ \varphi)'(z),0,\ldots ,0\right)^{tr},
\end{equation}
		which means that the linear combination of the columns of $J_{F_T}(\varphi(z))$ given by the entries of $\varphi'(z)$ is equal to the vector on the right side of the equation. Now, using observation (1), we compute $\det (J_{F_T}(x^*))$. Indeed, denoting by $J_{F_T}^\ell$ the $\ell$-th column of $J_{F_T}$,   Eq.~\eqref{eq:id} gives
		\[\varphi_i'(z)J_{F_T}^i(\varphi(z))= \left( 0,\ldots ,0,(F_{T,j}\circ \varphi)'(z),0,\ldots ,0\right)^{tr} -\sum_{k=1, k\neq i}^n \varphi_k'(z)J_{F_T}^k(\varphi(z)). \]
Let 
$\widehat{J}_{F_T}(x^*)$ be the matrix obtained by replacing the $i$-th column of $J_{F_T}(\varphi(z))$ by the vector
\[\left( 0,\ldots ,0,\frac{(F_{T,j}\circ \varphi)'(z)}{\varphi_i '(z)},0,\ldots ,0\right)^{tr}.\]
Then,  $\det (J_{F_T}(\varphi(z)))$ and $\det \big(\widehat{J}_{F_T}(\varphi(z)) \big)$ agree. 
Now, expanding the determinant of $\widehat{J}_{F_T}(\varphi(z)) $ along the $i$-th column gives
		\[\det (J_{F_T}(x^*))=\frac{(-1)^{i+j}}{\varphi_i '(z)}(F_{T,j}\circ \varphi)'(z) \det(J_{F_T}(x^*)_{J,I}).\]
This concludes the proof. 
\begin{flushright}
\qed
\end{flushright}

\subsection{Proof of Theorem~\ref{thm:bistability} and a lemma}
We start by proving Theorem~\ref{thm:bistability}. 
By the second hypothesis, the first $s-1$ Hurwitz determinants are positive, so the stability only depends on the sign of the last Hurwitz determinant, which as discussed before, agrees with the sign of $(-1)^s\det(J_{F_T}(\varphi(z)))$. According to Proposition~\ref{prop:reduce_polynomial}, it further  coincides with the sign of $\frac{(-1)^{s+i+j}}{\varphi_i'(z)}(F_{T,j}\circ \varphi)'(z)\det(J_{F_T}(\varphi(z))_{J,I})$. Since $\det(J_{F_T}(\varphi(z))_{J,I})$ has a constant sign for every $z\in \mathcal{E}$, the sign of the  last Hurwitz determinant changes when the sign of $(F_{T,j}\circ \varphi)'(z)$ does, and this is the derivative of a univariate differentiable  function whose real positive roots $z_1<\dots < z_\ell$ have multiplicity one and are ordered in an increasing way. Given that $(F_{T,j}\circ \varphi)$ is differentiable, the sign of its derivative evaluated at consecutive roots alternates, that is $(F_{T,j}\circ \varphi)(z_1)'>0,(F_{T,j}\circ \varphi)(z_3)'>0,\ldots$ and $(F_{T,j}\circ \varphi)(z_2)'<0,(F_{T,j}\circ \varphi)(z_4)'<0,\ldots $ or the other way around. 

In our setting this means that, once the sign of $(F_{T,j}\circ \varphi)'(z_k)$ is multiplied by $(-1)^{s+i+j}$ and by the sign of $\frac{1}{\varphi_i'(z)}\det(J_{F_T}(\varphi(z))_{J,I})$, either $\varphi(z_1),\varphi(z_3),\dots$ are  exponentially  stable and $\varphi(z_2),\varphi(z_4),\dots$ are unstable, or the other way around. In particular, if the sign of $(F_{T,j}\circ \varphi)(z_1)'$ times the sign of  $(-1)^{s+i+j} \frac{1}{\varphi_i'(z)}\det(J_{F_T}(\varphi(z))_{J,I})$ is positive, then $\varphi(z_1)$ is exponentially stable. 
\begin{flushright}
\qed
\end{flushright}

We include a simple lemma to ensure that, in the setting of Theorem~\ref{thm:bistability}, only the numerator of $F_{T,j}\circ \varphi$ needs to be considered.

\begin{Lemma}\label{lemma:rationalfunction}
	Under the assumptions of Theorem \ref{thm:bistability}, assume $( F_{T,j} \circ \varphi)(z)=\frac{a(z)}{b(z)}$ is a rational function in $z$ such that  $b(z)$ is positive in $\mathcal{E}$. Then the zeros of $( F_{T,j} \circ \varphi)(z)$  agree with the roots of $a(z)$ and the sign of  $( F_{T,j} \circ \varphi)'(z^*) $, and $a'(z^*)$ agree for all $z^*\in \mathcal{E}$ such that $a(z^*)=0$.
\end{Lemma}
\begin{proof}
The first part is straightforward, since the denominator of $( F_{T,j} \circ \varphi)(z)$ does not vanish in $\mathcal{E}$. 
For the second part, we have  
$F_{T,j} ( \varphi(z^*) )'=\frac{a'(z^*)b(z^*)-a(z^*)b'(z^*)}{b(z^*)^2} =\frac{a'(z^*)}{b(z^*)}$, and the conclusion follows from the fact that $b(z^*)>0$.
\end{proof}

\section{Computational challenges}\label{sec:comp_challenges}

In our context, the Hurwitz determinants are symbolic and depend on $\k$ and $x$ or $\xi$. 
Their computation requires the storage of functions with many terms, which easily becomes unfeasible in a regular PC. For example, for network (6) in Fig.~\ref{Figure 2}, $H_4$ and $H_5$ are polynomials in $\k$ and $x$ with respectively $1{,}732{,}192$ and $37{,}609{,}352$ monomials, before the evaluation at a parametrization.

For medium sized networks, some tricks can be applied under mass-action.  A first strategy is to parametrize the positive steady state variety using \emph{convex parameters} introduced by Clarke \cite{Clarke:1980tz,errami-hopf}. This conversion may reduce the number of parameters, mainly if the network has few reversible reactions. 

\smallskip
The second strategy corresponds to encode a monomial  $\eta\, x_1^{\alpha_1}\cdots x_n^{\alpha_n}$ as an $(n+1)$-tuple $(\eta,\alpha_1,\dots,\alpha_n)$, where $\eta$ is a rational function in $\k$.
We exploit relations among the $H_i$ obtained by expanding recursively along the last column of the submatrix of $H$ giving rise to $H_i$. For example, we have a  relation 
\[H_3= a_{s-3} H_2 - a_{s-1}(a_{s-4} H_1 + a_s a_{s-5})\] that holds for a generic polynomial. 
Assume $H_i$ is written as a sum of terms $H_{i,1}, \dots, H_{i,\ell}$ that can be computed.  Normally these terms arise from the coefficients of the characteristic polynomial and Hurwitz determinants $H_j$ with $j<i$, like the terms $a_{s-3} H_2$, and $a_{s-1}(a_{s-4} H_1 + a_s a_{s-5})$ for $H_3$. We identify the monomials of each term with their corresponding $(n+1)$-tuple and gather them into a list $L_1$.  Then, we create a list $L_2$ of the tuples of $L_1$ for which $\eta$ might not be a positive function of $\k$.
Note that there might be repeated monomials in $H_{i,1}, \dots, H_{i,\ell}$, whose coefficients we need to  add. To this end, for each element $(\eta^*,\alpha_1^*,\dots,\alpha_n^*)\in L_2$, we define the set
\[\bar{L}_{(\eta^*,\alpha_1^*,\ldots,\alpha_n^*)}\coloneqq. \big\{(\eta,\alpha_1,\dots,\alpha_n)\in L_1\, |\, \alpha_1=\alpha_1^*,\dots,\alpha_n=\alpha_n^* \big\}.\]
Then \[\bar{\eta}=\Sigma_{(\eta,\alpha_1,\dots,\alpha_n) \in \bar{L}_{(\eta^*,\alpha_1^*,\dots,\alpha_n^*)}} \ \eta \]
is the coefficient of $x_1^{\alpha^*_1}\cdots x_n^{\alpha^*_n}$ in $H_i$, and we inspect its sign. With this process we consider only the coefficients that might generate negative values of $H_i$ as they might be negative in some $H_{i,j}$. If all coefficients $\bar{\eta}$ are nonnegative, then so is $H_i$. 
This procedure requires substantially less memory, but it might take time as lists can be  long. With this strategy we have  determined the sign of $H_5$ in networks (5) and (6) in Fig.~\ref{Figure 2}. 

\smallskip
The third strategy applies to a special situation, namely when there exists a \emph{monomial positive parametrization}
\[\phi(\xi)_i = \beta_i\,  \xi_1^{b_{1i}}\cdots \xi_d^{b_{di}},\qquad i=1,\dots,n,\] with $\beta$ depending on the  rate constants \cite{PerezMillan}. 
Assume $H_i$ can be computed for $x\in \R^n_{>0}$, but evaluating at a parametrization and expanding the resulting polynomial to inspect its sign encounters memory issues. 
Let $B=(b_{ij}) \in \mathbb{Z}^{d\times n}$ be the matrix arising from the exponents of the monomial parametrization, and write
$(\xi^B)_j = \xi_1^{b_{1j}}\cdots \xi_d^{b_{dj}}$.

After evaluation at $\phi(\xi)$, a monomial $\eta\,  x_1^{\alpha_1}\cdots x_n^{\alpha_n}$ of $H_i$
becomes 
\[\eta \, \big(\beta_1\, \xi_1^{b_{11}}\cdots \xi_d^{b_{d1}}\big)^{\alpha_1}\cdots \big(\beta_n \, \xi_1^{b_{1n}}\cdots \xi_d^{b_{dn}}\big)^{\alpha_n}     = \eta\,  \beta_1^{\alpha_1} \cdots \beta_n^{\alpha_n}  \xi^{B\alpha}.\] 
Hence, to compute $H_i(\phi(\xi))$ we first write each monomial of $H_i(x)$ as an $(n+d+1)$-tuple given by $(\eta,\alpha_1,\ldots ,\alpha_n,0,\dots ,0)$. Then, we record the evaluation of each monomial with the $(n+d+1)$-tuple \[(\eta, \alpha_1,\dots,\alpha_n,(B\alpha)_1,\dots,(B\alpha)_d).\] 
This generates a list $L_1$ of $(n+d+1)$-tuples.  We define $L_2$ as the set of tuples of $L_1$ for which $\eta$ might not be a positive function of $\k$. Then, similar to the procedure explained above, for each $t=(\eta^*,\alpha_1^*,\dots,\alpha_n^*,(B\alpha)_1^*,\dots,(B\alpha)_d^*)\in L_2$, we consider all tuples that give rise to the same monomial in $\xi$, namely, we  define 
{\small
\[\bar{L}_{t}\coloneqq \big\{(\eta,\alpha_1,\ldots,\alpha_n,(B\alpha)_1,\ldots,(B\alpha)_d)\in L_1 \, | \, B\alpha=B\alpha^* \big\}.\] }
We compute now the coefficient of $\xi^{B\alpha}$, which is
\[ \bar{\eta}=\Sigma_{(\eta,\alpha_1,\ldots,\alpha_n,(B\alpha)_1,\ldots,(B\alpha)_d) \in \bar{L}_{t}} \ \eta\, \beta_1^{\alpha_1}\cdots \beta_n^{\alpha_n}. \]
We determine its sign to decide whether $H_i$ is positive.

With this approach, we computed the  Hurwitz determinants of networks (a), (b), (d) and (e) 
in Fig.~\ref{Figure 3}. However, $H_{s-1}$ is not positive and Theorem~\ref{thm:bistability} does not apply.

 \smallskip
To verify that a polynomial with both positive and negative coefficients attains both signs, one can study the associated Newton polytope, as employed in the context of reaction networks in \cite{FeliuPlos,Shiu-Hopf} to cite a few.
To assert that a polynomial only attains positive values despite having negative coefficients, one can employ techniques from  sum-of-squares \cite{Blekherman:SOS}  and polynomial optimization via sums of nonnegative circuit polynomials \cite{Dressler:Nullstellensatz,Iliman:Amoebas,feliu:2site}. Another approach, building on similar ideas but with a specific focus on reaction networks, was introduced in \cite{pantea-jac}.
However, the  size of the polynomials we encounter make these approaches  challenging.

\section*{Discussion}

 All the steps of our approach to determine the stability are symbolic, and therefore provide computer-assisted proofs for bistability. 
In the most favorable scenario where Theorem~\ref{thm:bistability} applies, the number of unstable and exponentially stable steady states is completely determined, and question (3) in the Introduction is answered.
In particular, if the reduced univariate equation has at least three solutions and the first steady state is exponentially stable,   the parameter region of bistability agrees with the parameter region giving three positive steady states. Finding the latter poses a simpler (though still hard) challenge, which can be (partially) addressed using recent methods \cite{FeliuPlos,dickenstein:regions}.

Under mass-action kinetics, reduction of the steady state equations to one polynomial  can in principle be achieved  using Gr\"obner bases and invoking the Shape Lemma \cite{cox:little:shea}. However, positivity is not addressed, and the interval $\mathcal{E}$ in Proposition~\ref{prop:reduce_polynomial} is rarely explicit. Reduction to one polynomial arises often after exploiting the inherent linearity the equations   have \cite{Fel_elim}. 

In our procedure, the Hurwitz criterion can be replaced by other criteria of algebraic nature, namely the Liénard-Chipart criterion in~\cite{Datta} or checking whether the matrix $Q_{x^*}$ is both a P-matrix and sign-symmetric. However, these criteria can only be used to assert exponential stability (see Appendix~\ref{app:criteria}, where these criteria are applied to the network in Eq.~\eqref{R.Ex}).

We have illustrated with numerous realistic examples that our approach determines bistability   after performing network reduction. To our knowledge, this is a new result for all networks in Fig.~\ref{Figure 3} but  network (d). For this one,  bistability was formally proven in \cite{rendall-2site} using methods from geometric singular perturbation theory and the accurate study of a reduced network. We see our approach as a big step towards the automatic detection of bistability in open parameter regions, which relies on purely algebraic manipulations instead of advanced analytic arguments.
Although the approach is applicable to arbitrary ODE systems, the special structure of the systems arising from reaction networks, specifically linearity, the existence of conservation laws and reduction techniques, make the approach particularly  suited for this scenario.

\medskip
\paragraph{{\bf Acknowledgments. }}
The authors acknowledge funding from the Independent Research Fund of Denmark. We thank L. Brustenga, A. Dickenstein, B. Pascual Escudero, A. Sadeghimanesh, and C. Wiuf for useful comments in earlier versions of this manuscript.


\begin{thebibliography}{10}

\bibitem{Ali-Angeli}
M.~Ali Al-Radhawi and D.~Angeli.
\newblock {New approach to the stability of chemical reaction networks:
  Piecewise linear in rates lyapunov functions}.
\newblock {\em IEEE Trans. Automat. Control}, 61:76--89, 2016.

\bibitem{Anagnost}
J.~J. Anagnost and C.~A. Desoer.
\newblock An elementary proof of the {R}outh-{H}urwitz criterion.
\newblock {\em Circuits systems Signal process}, 10(1):101--114, 1991.

\bibitem{angelisontag2}
D.~Angeli, P.~De~Leenheer, and E.~Sontag.
\newblock {{G}raph-theoretic characterizations of monotonicity of chemical
  networks in reaction coordinates}.
\newblock {\em J. Math. Biol.}, 61:581--616, 2010.

\bibitem{Banaji-Pantea}
M.~Banaji and C.~Pantea.
\newblock {Some results on injectivity and multistationarity in chemical
  reaction networks}.
\newblock {\em SIAM J. Appl. Dyn. Syst.}, 15(2):807--869, 2016.

\bibitem{banaji-pantea-inheritance}
M.~Banaji and C.~Pantea.
\newblock The inheritance of nondegenerate multistationarity in chemical
  reaction networks.
\newblock {\em SIAM J. Appl. Math.}, 78(2):1105--1130, 2018.

\bibitem{Barnett:hurwitz-routh}
S.~Barnett.
\newblock A new formulation of the theorems of {H}urwitz, {R}outh and {S}turm.
\newblock {\em J. Inst. Math. Appl.}, 8:240--250, 1971.

\bibitem{basu:book}
S.~Basu, R.~Pollack, and M.-F. Roy.
\newblock {\em Algorithms in real algebraic geometry}, volume~10 of {\em
  Algorithms and Computation in Mathematics}.
\newblock Springer-Verlag, Berlin, second edition, 2006.

\bibitem{dickenstein:regions}
F.~Bihan, A.~Dickenstein, and M.~Giaroli.
\newblock Lower bounds for positive roots and regions of multistationarity in
  chemical reaction networks.
\newblock {\em J. Algebra}, 542:367--411, 2020.

\bibitem{Blekherman:SOS}
G.~Blekherman.
\newblock Nonnegative polynomials and sums of squares.
\newblock {\em J. Am. Math. Soc.}, 25(3):617--635, 2012.

\bibitem{inflows}
D.~Cappelletti, E.~Feliu, and C.~Wiuf.
\newblock Addition of flow reactions preserving multistationarity and
  bistability.
\newblock {\em Math. Biosci.}, 320:108295, 2020.

\bibitem{Carlson}
D.~Carlson.
\newblock {A class of positive stable matrices}.
\newblock {\em J. Res. Nat. Bur. Standards Sect B.}, 78.B:1--2, 1974.

\bibitem{Clarke:1980tz}
B.~L. Clarke.
\newblock {\em {Stability of Complex Reaction Networks}}, volume~43 of {\em
  Advances in Chemical Physics}.
\newblock John Wiley {\&} Sons, Inc., Hoboken, NJ, USA, 1980.

\bibitem{Clarke:graph}
Bruce~L. Clarke.
\newblock Graph theoretic approach to the stability analysis of steady state
  chemical reaction networks.
\newblock {\em The Journal of Chemical Physics}, 60(4):1481--1492, 1974.

\bibitem{FeliuPlos}
C.~Conradi, E.~Feliu, M.~Mincheva, and C.~Wiuf.
\newblock Identifying parameter regions for multistationarity.
\newblock {\em PLoS Comput. Biol.}, 13(10):e1005751, 2017.

\bibitem{conradi-switch}
C.~Conradi and D.~Flockerzi.
\newblock Switching in mass action networks based on linear inequalities.
\newblock {\em SIAM J. Appl. Dyn. Syst.}, 11(1):110--134, 2012.

\bibitem{conradi-mincheva}
C.~Conradi and M.~Mincheva.
\newblock Catalytic constants enable the emergence of bistability in dual
  phosphorylation.
\newblock {\em J. R. S. Interface}, 11(95), 2014.

\bibitem{Shiu-Hopf}
C.~Conradi, M.~Mincheva, and A.~Shiu.
\newblock Emergence of oscillations in a mixed-mechanism phosphorylation
  system.
\newblock {\em Bull. Math. Biol.}, 81(6):1829--1852, 2019.

\bibitem{cox:little:shea}
D.~Cox, J.~Little, and D.~O'Shea.
\newblock {\em Ideals, varieties, and algorithms}.
\newblock Undergraduate Texts in Mathematics. Springer, New York, third
  edition, 2007.

\bibitem{Craciun-Sturmfels}
G.~Craciun, A.~Dickenstein, A.~Shiu, and B.~Sturmfels.
\newblock Toric dynamical systems.
\newblock {\em J. Symbolic Comput.}, 44(11):1551--1565, 2009.

\bibitem{craciun-feinberg}
G.~Craciun and M.~Feinberg.
\newblock {{M}ultiple equilibria in complex chemical reaction networks:
  extensions to entrapped species models}.
\newblock {\em Syst. Biol. (Stevenage)}, 153:179--186, 2006.

\bibitem{Datta}
B.~N. Datta.
\newblock {An elementary proof of the stability criterion of Liénard and
  Chipart}.
\newblock {\em Linear Algebra Appl.}, 22:89--96, 1978.

\bibitem{Dickenstein:2011p1112}
A.~Dickenstein and M.~P{\'e}rez~Mill{\'a}n.
\newblock How far is complex balancing from detailed balancing?
\newblock {\em Bull. Math. Biol.}, 73(4):811--828, 2011.

\bibitem{banaji-donnell-II}
P.~Donnell and M.~Banaji.
\newblock Local and global stability of equilibria for a class of chemical
  reaction networks.
\newblock {\em SIAM J. Appl. Dyn. Syst.}, 12(2):899--920, 2013.

\bibitem{control}
P.~Donnell, M.~Banaji, A.~Marginean, and C.~Pantea.
\newblock Control: an open source framework for the analysis of chemical
  reaction networks.
\newblock {\em Bioinformatics}, 30(11):1633-34, 2014.

\bibitem{Dressler:Nullstellensatz}
M.~Dressler, S.~Iliman, and T.~de~Wolff.
\newblock A positivstellensatz for sums of nonnegative circuit polynomials.
\newblock {\em SIAM J. Appl. Alg. Geom.}, 1(1):536--555, 2017.

\bibitem{crnttoolbox}
P.~Ellison, M.~Feinberg, H.~Ji, and D.~Knight.
\newblock Chemical reaction network toolbox, version 2.2.
\newblock Available online at \url{http://www.crnt.osu.edu/CRNTWin}, 2012.

\bibitem{errami-hopf}
H.~Errami, M.~Eiswirth, D.~Grigoriev, W.~M. Seiler, T.~Sturm, and A.~Weber.
\newblock Detection of {H}opf bifurcations in chemical reaction networks using
  convex coordinates.
\newblock {\em J. Comput. Phys.}, 291:279 -- 302, 2015.

\bibitem{Feinberg1972}
M.~Feinberg.
\newblock Complex balancing in general kinetic systems.
\newblock {\em Arch. Rational Mech. Anal.}, 49(3):187--194, 1972.

\bibitem{feinberg-book}
M.~Feinberg.
\newblock {\em Foundations of Chemical Reaction Network Theory}, volume 202 of
  {\em Applied Mathematical Sciences}.
\newblock Springer International Publishing, 2019.

\bibitem{feliu:2site}
E.~Feliu, N.~Kaihnsa, T.~de~Wolff, and O.~Yürük.
\newblock The kinetic space of multistationarity in dual phosphorylation.
\newblock {\em J. Dyn. Differ. Equ.}, To appear, 2020.

\bibitem{rendall-feliu-wiuf}
E.~Feliu, A.~D. Rendall, and C.~Wiuf.
\newblock A proof of unlimited multistability for phosphorylation cycles.
\newblock {\em Nonlinearity}, To appear, 2020.

\bibitem{Feliu:royal}
E.~Feliu and C.~Wiuf.
\newblock Enzyme-sharing as a cause of multi-stationarity in signalling
  systems.
\newblock {\em J. R. S. Interface}, 9(71):1224--32, 2012.

\bibitem{Fel_elim}
E.~Feliu and C.~Wiuf.
\newblock Variable elimination in chemical reaction networks with mass-action
  kinetics.
\newblock {\em SIAM J. Appl. Math.}, 72:959--981, 2012.

\bibitem{feliu:intermediates}
E.~Feliu and C.~Wiuf.
\newblock Simplifying biochemical models with intermediate species.
\newblock {\em J. R. S. Interface}, 10:20130484, 2013.

\bibitem{fwptm}
E.~Feliu and C.~Wiuf.
\newblock Variable elimination in post-translational modification reaction
  networks with mass-action kinetics.
\newblock {\em J. Math. Biol.}, 66(1):281--310, 2013.

\bibitem{feng:allosteric}
S.~Feng, M.~S\'aez, C.~Wiuf, E.~Feliu, and O.~S. Soyer.
\newblock {C}ore signalling motif displaying multistability through multi-state
  enzymes.
\newblock {\em J. R. S. Interface}, 13(123), 2016.

\bibitem{Guidi:Bistability}
G.~M. Guidi and A.~Goldbeter.
\newblock Bistability without hysteresis in chemical reaction systems: A
  theoretical analysis of irreversible transitions between multiple steady
  states.
\newblock {\em J. Phys. Chem. A}, 101(49):9367--9376, 1997.

\bibitem{gunawardena-linear}
J.~Gunawardena.
\newblock {A linear framework for time-scale separation in nonlinear
  biochemical systems}.
\newblock {\em PLoS ONE}, 7(5):e36321, 2012.

\bibitem{rendall-2site}
J.~Hell and A.~D. Rendall.
\newblock A proof of bistability for the dual futile cycle.
\newblock {\em Nonlinear Anal. Real World Appl.}, 24:175--189, 2015.

\bibitem{hornjackson}
F.~J.~M. Horn and R.~Jackson.
\newblock General mass action kinetics.
\newblock {\em Arch. Rational Mech. Anal.}, 47:81--116, 1972.

\bibitem{Iliman:Amoebas}
S.~Iliman and T.~de~Wolff.
\newblock Amoebas, nonnegative polynomials and sums of squares supported on
  circuits.
\newblock {\em Research in the Mathematical Sciences}, 3(1):9, Mar 2016.

\bibitem{joshi-shiu-II}
B.~Joshi and A.~Shiu.
\newblock Atoms of multistationarity in chemical reaction networks.
\newblock {\em J. Math. Chem.}, 51(1):153--178, 2013.

\bibitem{joshi-shiu-III}
B.~Joshi and A.~Shiu.
\newblock A survey of methods for deciding whether a reaction network is
  multistationary.
\newblock {\em Math. Model. Nat. Phenom.}, 10(5):47--67, 2015.

\bibitem{feliu:unlimited}
V.~B. Kothamachu, E.~Feliu, L.~Cardelli, and O.~S. Soyer.
\newblock Unlimited multistability and boolean logic in microbial signalling.
\newblock {\em J. R. S. Interface}, 12:20150234, 2015.

\bibitem{MullerSigns}
S.~M\"uller, E.~Feliu, G.~Regensburger, C.~Conradi, A.~Shiu, and
  A.~Dickenstein.
\newblock Sign conditions for the injectivity of polynomial maps in chemical
  kinetics and real algebraic geometry.
\newblock {\em Found. Comput. Math.}, 16:69--97, 2016.

\bibitem{muller}
S.~M\"uller and G.~Regensburger.
\newblock Generalized mass action systems: {C}omplex balancing equilibria and
  sign vectors of the stoichiometric and kinetic-order subspaces.
\newblock {\em SIAM J. Appl. Math.}, 72:1926--1947, 2012.

\bibitem{Ninfape20}
A.~J. Ninfa and A.~E. Mayo.
\newblock Hysteresis vs. graded responses: The connections make all the
  difference.
\newblock {\em Science Signaling}, 2004(232):20--20, 2004.

\bibitem{Ozbudak:BistabilityEcoli}
E.~M. Ozbudak, M.~Thattai, H.~N. Lim, B.~I. Shraiman, and A.~van Oudenaarden.
\newblock Multistability in the lactose utilization network of escherichia
  coli.
\newblock {\em Nature}, 427:737, 2004.

\bibitem{pantea-jac}
C.~Pantea, H.~Koeppl, and G.~Craciun.
\newblock Global injectivity and multiple equilibria in uni- and bi-molecular
  reaction networks.
\newblock {\em Discrete Contin. Dyn. Syst. Ser. B}, 17(6):2153--2170, 2012.

\bibitem{Dickenstein-MESSI}
M.~P\'{e}rez~Mill\'{a}n and A.~Dickenstein.
\newblock The structure of {MESSI} biological systems.
\newblock {\em SIAM J. Appl. Dyn. Syst.}, 17:1650--1682, 2018.

\bibitem{PerezMillan}
M.~P{\'e}rez~Mill{\'a}n, A.~Dickenstein, A.~Shiu, and C.~Conradi.
\newblock Chemical reaction systems with toric steady states.
\newblock {\em Bull. Math. Biol.}, 74:1027--1065, 2012.

\bibitem{perko}
L.~Perko.
\newblock {\em Differential equations and dynamical systems}, volume~7 of {\em
  Texts in Applied Mathematics}.
\newblock Springer-Verlag, New York, third edition, 2001.

\bibitem{amir_multi}
A.~Sadeghimanesh and E.~Feliu.
\newblock The multistationarity structure of networks with intermediates and a
  binomial core network.
\newblock {\em Bull. Math. Biol.}, 81:2428–2462, 2019.

\bibitem{saez_elim}
M.~Sáez, C.~Wiuf, and E.~Feliu.
\newblock Nonnegative linear elimination for chemical reaction networks.
\newblock {\em SIAM J. Appl. Math.}, 79(6):2434–2455, 2019.

\bibitem{TG-rational}
M.~Thomson and J.~Gunawardena.
\newblock {{T}he rational parameterization theorem for multisite
  post-translational modification systems}.
\newblock {\em J. Theor. Biol.}, 261:626--636, 2009.

\bibitem{volpert}
A.~I. Vol'pert.
\newblock Differential equations on graphs.
\newblock {\em Math. {USSR}-{S}b}, 17:571–582, 1972.

\bibitem{wiuf-feliu}
C.~Wiuf and E.~Feliu.
\newblock Power-law kinetics and determinant criteria for the preclusion of
  multistationarity in networks of interacting species.
\newblock {\em SIAM J. Appl. Dyn. Syst.}, 12:1685--1721, 2013.

\bibitem{yang-hopf}
X.~Yang.
\newblock Generalized form of {H}urwitz-{R}outh criterion and {H}opf
  bifurcation of higher order.
\newblock {\em Appl. Math. Lett.}, 15(5):615--621, 2002.

\end{thebibliography}

\appendix
\section*{Appendices}
In Appendix~\ref{app:criteria} we expand on other stability criteria that could be used in our procedure instead of the Hurwitz criterion, and in Appendix~\ref{app:examples} we provide details of the analysis of the networks in Figures~\ref{Figure 2} and \ref{Figure 3}.

\section{Other algebraic criteria for stability}\label{app:criteria}
In this appendix we discuss two additional criteria to decide the stability of the steady states. 
Similarly to the Hurwitz criterion, the Liénard-Chipart criterion in~\cite{Datta}, determines whether all the roots of a polynomial have negative real part, and requires a smaller amount of computations than the Hurwitz criterion. Before introducing the criterion, we need some ingredients. 

\begin{definition}
	\begin{itemize}
		\item The \textit{Bezout matrix} of two polynomials $h(x)=h_nx^n+h_{n-1}x^{n-1}+\cdots +h_{1}x+h_0$ and $g(x)=g_mx^m+g_{m-1}x^{m-1}+\cdots +g_{1}x+g_0$ with $n\geq m$, denoted by $B_{h,g}$, is defined as the representation matrix of the bilinear form 
\[B(h,g;x,y)=\frac{h(x)g(y)-h(y)g(x)}{x-y}=\sum_{i,j=0}^{n-1}b_{ik}x^iy^k,\]
that is, $B_{h,g}\coloneqq (b_{ik})$. This is a symmetric matrix of size $n\times n$.
		
		\item A square matrix $A\in \R^{n\times n}$ is called a P-matrix if all its principals minors are positive, that is, $\det(A_{I,I})>0$ for every subset $I\subseteq \{1,\ldots ,n\}$. If $A$ is symmetric, this is equivalent to $A$ being positive definite. 
	\end{itemize}
\end{definition}
We now present the first additional stability criterion.
\begin{criterion}[Liénard-Chipart]
	All the roots of a polynomial $p(x)=x^s+p_{s-1}x^{s-1}+\ldots +p_{1}x+p_0$ with $p_i\in \mathbb{R}$ and $p_0\neq 0$ have negative real part if and only if, after writing $p(x)=h(x^2)+xg(x^2)$, the Bezout matrix  $B_{h,g}$ of $h$ and $g$ is positive definite and $p_i>0$ for $i=1,\ldots ,s$.
\end{criterion}

In this criterion the polynomials $h$ and $g$ are associated with the even and odd parts of $p$ respectively. Note that the degrees of $h$ and $g$ are at most $\lfloor\frac{s}{2}\rfloor$, therefore the size of $B_{h,g}$ is $\lfloor \frac{s}{2} \rfloor$. Additionally, since $B_{h,g}$ is symmetric, $B_{h,g}$ is positive definite if and only if it is a P-matrix. 
 
Unlike the Hurwitz criterion, Liénard-Chipart does not give a result regarding instability. If $B_{h,g}$ is not a P-matrix, it is not possible to determine whether the eigenvalues that do not have negative real part, have positive or zero real part.

\medskip
In order to apply the criterion to Example~\ref{ex:HK}, we write $q_{x}(\lambda)$ as 
$h(\lambda^2)+\lambda g(\lambda^2)$, with
	\[h(\lambda)=\lambda+\k_1\k_2x_2+\k_2\k_3x_3+\k_1\k_3 \quad \mbox{ and } \quad g(\lambda)=\k_2x_2+\k_2x_3+\k_1+\k_3.\]
The Bezout matrix is then $B_{h,g}=(x_2+x_3)\k_2+\k_1+\k_3$, which is clearly a P-matrix. 
By the Liénard-Chipart criterion, we conclude that all the roots of $q_{x}$ have negative real part, recovering thereby that the only positive steady state in each stoichiometric compatibility class is exponentially stable.

The Liénard-Chipart criterion is the most efficient if we consider the amount of determinants that have to be computed to reach a decision. Given that the Bezout matrix is symmetric, to check that is positive definite, it is only necessary to compute the principal minors and check whether they are positive. Thus the required amount of determinants  is at most $\sum_{i=0}^{\lfloor\frac{n+1}{2}\rfloor} {\binom{\lfloor\frac{n+1}{2}\rfloor}{i}}$. Although this criteria computes the smallest amount of determinants, for some examples, the entries of the Bezout matrix are larger than the entries of the Jacobian. In those cases the memory of a regular PC is still not enough to store the computations.

\medskip
We conclude the list of stability criteria of algebraic nature with one more criterion, which does not rely on the computation of the characteristic polynomial. For a square matrix $A$, we say that 
$A$ is \textit{sign symmetric} if $\det(A_{I,J})\det (A_{J,I})\geq 0$ for every $I,J\subset \{1,\ldots ,n\}$ with the same cardinality.

\begin{criterion}[P-matrices that are sign symmetric]\label{cri:P}
	If a square matrix $A$ is both a P-matrix and sign symmetric, then all its eigenvalues have positive real part. 
\end{criterion}
With this criterion, proved in~\cite{Carlson}, if $-A$ is a P-matrix and sign-symmetric, then all its eigenvalues have negative real part. For reaction networks, we apply the criterion to the matrix $-Q_{x}$. In Example~\ref{ex:HK}, 
we compute the minors of size 1 and 2 of $A=-Q_{x}$. The minors of size $1$ are the entries of the matrix, which are all positive. The only minor of size 2 is \[ \det(A)=\k_1\k_2x_2+\k_2\k_3x_3+\k_1\k_3.\]
	All minors are polynomials that are positive for all $x\in\mathbb{R}^4_{>0}$ and positive $\k$. Therefore $-Q_{x}$ is a P-matrix and sign symmetric. Hence,  with this new criterion we recover the conclusion that the only positive steady state is exponentially stable.

\medskip
\medskip
While the Hurwitz and Liénard-Chipart criteria are applied to the characteristic polynomial of $L J_v(x^*)R_0$, which is independent of the choice of $R_0$, Criterion~\ref{cri:P} is applied directly to the matrix $-L J_v(x^*)R_0$. 
By Proposition~\ref{prop:basic_jacobian}, two different choices of $R_0$ give rise to two similar matrices. However,  the properties of being P-matrix and sign symmetric are not preserved on similar matrices. As a small example consider 	  	
	  	\[ A=\left(\begin{array}{cc}
	  		2 & 1 \\
	  		3 & 4
	  	\end{array}\right) \qquad \textrm{ and }\qquad B=\left(\begin{array}{cc}
	  		-1 & -4 \\
	  		3 & 7
	  	\end{array}\right); \] these matrices are similar through 
	  	{\small $M=\left(\begin{array}{cc}
	  		1 & 1 \\
	  		0 & 1
	  	\end{array}\right)$}, %
 but $A$ is both a P-matrix and sign symmetric and $B$ is neither. 

A comparison of the amount of operations required for each stability criterion shows that deciding whether a matrix is a P-matrix and sign-symmetric requires the largest amount of operations. In this case, all $\sum_{i=1}^s {\binom{s}{i}}^2$ minors of the matrix must be computed, which requires the storage of a large amount of information given that the entries of the Jacobian are typically polynomials. 

\section{Examples}\label{app:examples} 
In this section we provide extra details on the study of the networks in Figures~\ref{Figure 2} and \ref{Figure 3}. For more details on the computations, we refer the reader to the accompanying \texttt{Maple}  file.
Concentrations are denoted by corresponding  lower case letters: the concentration of species $X_i$ is denoted by $x_i$. 

\subsection{Two substrate enzyme catalysis (Fig.~\ref{Figure 1})}
We consider  the following network with mass-action kinetics
{\small
\begin{center}
\begin{tabular}{ccccc}
\ce{E + S1 <=>[\k_1][\k_2] ES1} & \ce{E + S2 <=>[\k_3][\k_4] ES2} & \ce{S2 + ES1 <=>[\k_5][\k_6] ES1S2}\\[0.2cm]
 & \ce{ES1S2 <=>[\k_7][\k_8] E + P} & \ce{S1 + ES2 <=>[\k_9][\k_{10}] ES1S2}.
\end{tabular}
\end{center}}
		
This network consists of an enzyme $E$ that binds two substrates, $S_1$ and $S_2$, in order to catalyze the reversible conversion to the product $P$. The binding is unordered. It was proven in~\cite{FeliuPlos} that this network has a unique steady state in each stoichiometric compatibility class for every set of reaction rate constants. We now prove that this steady state is exponentially stable. First, denote the species as $X_1$=\ce{E}, $X_2$=\ce{S1}, $X_3$=\ce{ES1}, $X_4$=\ce{S2}, $X_5$=\ce{ES2}, $X_6$=\ce{ES1S2} and $X_7$=\ce{P}.  The ODE system is
{\small \begin{align*}
\tfrac{dx_1}{dt}  &= -\k_1x_1x_2-\k_3x_1x_4-\k_{10}x_1x_7+\k_2x_3+\k_4x_5+\k_9x_6  & \tfrac{dx_5}{dt} &= \k_3x_1x_4-\k_8x_2x_5-\k_4x_5+\k_7x_6 \\
			\tfrac{dx_2}{dt} &= -\k_1x_1x_2-\k_8x_2x_5+\k_2x_3+\k_7x_6  & \tfrac{dx_6}{dt} &= \k_5x_3x_4+\k_8x_2x_5+\k_{10}x_1x_7 \\
			\tfrac{dx_3}{dt} &=  \k_1x_1x_2-\k_5x_3x_4-\k_2x_3+\k_6x_6 &  & \quad -\k_6x_6-\k_7x_6-\k_9x_6 \\
			\tfrac{dx_4}{dt} &=  -\k_3x_1x_4-\k_5x_3x_4+\k_4x_5+\k_6x_6 &   \tfrac{dx_7}{dt} &=-\k_{10}x_1x_7+\k_9x_6,
\end{align*}}
and the conservation laws  are
{\small
\[ x_1+x_3+x_5+x_6=T_1, \qquad x_2+x_3+x_6+x_7=T_2\qquad \mbox{ and } \qquad x_4+x_5+x_6+x_7=T_3. \]}
With this choice of conservation laws we have 
{\small \[ F_T(x)= \left( \begin{array}{c}
			 x_1+x_5+x_5+x_6-T_1\\
			 x_2+x_3+x_6+x_7-T_2\\
			 \k_1x_1x_2-\k_5x_3x_4-\k_2x_3+\k_6x_6  \\
			 x_4+x_5+x_6+x_7-T_3 \\
			 \k_3x_1x_4-\k_8x_2x_5-\k_4x_5+\k_7x_6\\
			 \k_5x_3x_4+\k_8x_2x_5+\k_{10}x_1x_7-\k_6x_6-\k_7x_6-\k_9x_6 \\
			  -\k_{10}x_1x_7+\k_9x_6
\end{array} \right). \]}%
Here $s=4$.
We compute $q_x$ and the Hurwitz determinants in {\tt Maple}, and find that all but the last have all coefficients positive, and thus are positive. We find next a positive parametrization by solving the steady state equations in the variables $x_3,x_5$, $x_6$, $x_7$ following the methods proposed in \cite{Fel_elim,FeliuPlos}: 
	{\footnotesize 	\begin{align*}
x_3 & =\frac{ x_1x_2(\k_6\k_8(\k_1x_2+\k_3x_4)+\k_1\k_4(\k_6+\k_7))}{\k_2\k_6\k_8x_2+\k_4\k_5\k_7x_4+\k_2\k_4\k_6+\k_2\k_4\k_7}	&
x_5 & = \frac{x_1x_4( \k_5\k_7(\k_1x_2+\k_3x_4)+\k_2\k_3(\k_6+\k_7))}{\k_2\k_6\k_8x_2+\k_4\k_5\k_7x_4+\k_2\k_4\k_6+\k_2\k_4\k_7} \\
x_6 & = \frac{x_1x_2x_4(\k_5\k_8(\k_1x_2+\k_3x_4)+\k_1\k_4\k_5+\k_2\k_3\k_8)}{\k_2\k_6\k_8x_2+\k_4\k_5\k_7x_4+\k_2\k_4\k_6+\k_2\k_4\k_7}  &
x_7& = \frac{x_2x_4\k_9( \k_5\k_8(\k_1x_2+\k_3x_4)+\k_1\k_4\k_5+\k_2\k_3\k_8)}{(\k_2\k_6\k_8x_2+\k_4\k_5\k_7x_4+\k_2\k_4\k_6+\k_2\k_4\k_7)\k_{10}}.
\end{align*}}%

After evaluation of the independent term of $q_x$ at the parametrization, $H_4$ becomes positive. 
We conclude that for any choice of reaction rate constants, the network for two substrate enzyme catalysis has exactly one positive steady state in each stoichiometric compatibility class, which is exponentially stable.

\subsection{Gene transcription network (Fig.~\ref{Figure 1})}

We consider  the following  network:
\begin{center}
\begin{tabular}{cccc}
\ce{X1 ->[\k_1] X1 + P1} & \ce{X2 ->[\k_2] X2 + P2} & \ce{P1 ->[\k_3] 0} & \ce{P2 ->[\k_4] 0} \\[8pt]
\ce{X2 + P1 <=>[\k_5][\k_6] X2P1} & \ce{2P2 <=>[\k_7][\k_8] P2P2} &
\ce{X1 + P2P2 <=>[\k_9][\k_{10}] X1P2P2}.
\end{tabular}
\end{center}
We denote the species as $X_1$ =\ce{ X_1}, $X_2$ =\ce{X2}, $X_3$=\ce{ P_1}, $X_4 $= \ce{P2}, $X_5$ = \ce{X_2P_1}, $X_6$ = \ce{P_2P_2}, and $X_7 = X_1P_2P_2$. Additionally, we are under the assumption of mass-action kinetics. It was proven in~\cite{FeliuPlos} that for each set of positive reaction rate constants there is a stoichiometric compatibility class that contains at least two positive steady states.
The ODE system is
	{\small \begin{align*}
\tfrac{dx_1}{dt} &= -\k_{9}x_{1}x_{6}+\k_{10}x_{7}  & \tfrac{dx_2}{dt} & =-\k_{5}x_{2}x_{3}+\k_{6}x_{5} \\
\tfrac{dx_3}{dt} &=-\k_{5}x_{2}x_{3}+\k_{1}x_{1}-\k_{3}x_{3}+\k_{6}x_{5} & \tfrac{dx_4}{dt} &=  -2\,\k_{7}x_{4}^{2}+\k_{2}x_{2}-\k_{4}x_{4}+2\,\k_{8}x_{6}\\ 
\tfrac{dx_5}{dt} &= \k_{5}x_{2}x_{3}-\k_{6}x_{5} & \tfrac{dx_6}{dt} &= \k_{7}x_{4}^{2}-\k_{9}x_{1}x_{6}-\k_{8}x_{6}+\k_{10}x_{7}\\ 
\tfrac{dx_7}{dt} &= \k_{9}x_{1}x_{6}-\k_{10}x_{7},
		\end{align*}}%
and the conservation laws are  $x_1+x_7=T_1$ and $x_2+x_5=T_2$.
These give rise to the function $F_T(x)$:
{\small \[ F_T(x)= \left( \begin{array}{c}
	x_1+x_7-T_1\\
	x_2+x_5-T_2\\
	-\k_5x_2x_3+\k_1x_1-\k_3x_3+\k_6x_5  \\
	-2\k_7x_4^2+\k_2x_2-\k_4x_4+2\k_8x_6\\
	 \k_5x_2x_3-\k_6x_5 \\
	 \k_7x_4^2-\k_9x_1x_6-\k_8x_6+\k_{10}x_7 \\
	 \k_9x_1x_6-\k_{10}x_7
\end{array} \right). \]}%

Here $s=5$. We find $q_x$ and compute the 5 Hurwitz determinants for $x\in \R^7_{>0}$. 
We find that $H_3,H_4,H_5$ have coefficients of both signs. We proceed to find a parametrization by solving the steady state equations in $x_3,\dots,x_7$, which gives:
{\small \[ x_{{3}}={\frac {\k_{{1}}x_{{1}}}{\k_{{3}}}},\quad x_{{4}}={\frac {\k_{{2}}x_{{2}}}{\k_{{4}}}},
\quad x_{{5}}={\frac {\k_{{1}}\k_{{5}}x_{{1}}x_{{2}}}{\k_{{3}}\k_{{6}}}},
\quad x_{{6}}={\frac {\k_2^{2}\k_{{7}}x_2^2}{\k_4^{2}\k_{{8}}}},
\quad x_{{7}}=\frac {\k_ 2^{2}\k_7\k_{{9}}x_{{1}}x_2^{2}}{\k_4^{2}\k_8\k_{10}}.\] }%
After evaluating $H_3,H_4$ and $H_5$ in this parametrization, $H_3$ and $H_4$ become rational functions in $x_1,x_2$ and $\k$ with all coefficients positive. Hence they are positive as well.

This means that the stability of the steady state is determined by the sign of $H_5$. We check whether we can apply Theorem~\ref{thm:bistability}. By solving $F_T(x)=0$ in $x_2,\dots,x_7$ after removal of the conservation law with total amount $T_1$, we obtain 
{\small \begin{align*}
x_{{2}} & ={\frac {\k_{{3}}\k_{{6}}T_{{2}}}{\k_{{1}}\k_{{5}}x_{{1}}+
\k_{{3}}\k_{{6}}}}, & x_{{3}}& ={\frac {\k_{{1}}x_{{1}}}{\k_{{3}}}}, & x_{{4}}& ={\frac {\k_{2}\k_{3}\k_{6}T_{2}}{ \left( \k_{{1}}\k_{{5}}x_{{1}}+\k_{{3}}\k_{{6}} \right) \k_{{4}}}},
\\ x_{{5}} & ={\frac {\k_{{1}}\k_{{5}}x_{{1}}T_{{2}}}{\k_{{1}}\k_{{5}}x_{{1}}+\k_{{3}}\k_{{6}}}}, & 
x_{{6}}& ={\frac {\k_2^{2}\k_3^{2}\k_6^{2}\k_7T_2^{2}}{ ( \k_{{1}}\k_{{5}}x_{{1}}+\k_{{3}}\k_{{6}})^{2}\k_4^{2}\k_{{8}}}}, & 
x_{{7}}& ={\frac {\k_2^{2}\k_3^{2}\k_6^{2}\k_7\k_9x_1T_2^{2}}{ \left( \k_{{1}}\k_{{5}}x_{{1}}+\k_{{3}}\k_{{6}} \right) ^{2}\k_{{4}}^{2}\k_{{8}}\k_{{10}}}}. 
 \end{align*} }%
These expressions define $\varphi$, with $z=x_1$, $\varphi_1(z)=z$ and $\mathcal{E}=\R_{>0}$. By inserting these expressions into the conservation law with $T_1$, we conclude that the solutions of $F_T(x)=0$ are in one to one correspondence with the zeroes of the function 
{\small \begin{equation}
	\begin{split}
		(F_{T,1}\circ\varphi)(z) &=\frac{1}{(\k_1\k_5z+\k_3\k_6)^2\k_4^2\k_8\k_{10}} \left[  \k_1^2\k_4^2\k_5^2\k_8\k_{10}z^3+(-T_1\k_1\k_5 +2\k_3\k_6)\k_1\k_4^2\k_5\k_8\k_{10}z^2 +\right.  \\ 
		 (&T_2^2\k_2^2\k_3^2\k_6^2\k_7\k_9-2T_1\k_1\k_3\k_4^2\k_5\k_6\k_8\k_{10}+ \left. \k_3^2\k_4^2\k_6^2\k_8\k_{10})z-T_1\k_3^2\k_4^2\k_6^2\k_8\k_{10} \right].   \nonumber
	\end{split}
	\end{equation}}%
 The numerator of this function has degree three in $z$, so using Lemma~\ref{lemma:rationalfunction}, the maximum number of positive steady states in each stoichiometric compatibility class is three.
The first hypothesis in Theorem~\ref{thm:bistability} is satisfied, as
\[ \tfrac{1}{\varphi_i'(z)}\det(J_{F_T}(\varphi(z))_{J,I}) = \det(J_{F_T}(\varphi(z))_{\{2,\dots,7\},\{2,\dots,7\}}) = -( \k_1\k_5x_1+\k_3\k_6 ) \k_4\k_8\k_{10} <0.
  \]
By Theorem~\ref{thm:bistability}, the sign of \eqref{eq:importantsign} 
is positive since $s+i+j=7$ and the sign of  $(F_{T,1}\circ \varphi)'(z_1) $ is positive as the independent term of the numerator of $(F_{T,1}\circ \varphi)(z_1) $ is negative.
Therefore, the stability of the steady states alternate with $z$ starting with an exponentially stable steady state. Specifically, if a stoichiometric compatibility class has one positive steady state, then it is exponentially stable. If it has three positive steady states, then two of them are exponentially stable and the other one is unstable.	Bistability is guaranteed whenever the network has three positive steady states.

\subsection{Monostationary networks from Fig.~\ref{Figure 2}}

Networks (1) to (4) are straightforward to analyze, since all coefficients of the Hurwitz determinants  in $x$ and $\k$ are positive; hence the Hurwitz determinants are positive for all $x\in \R^n_{>0}$ and $\k\in \R^m_{>0}$.  

\medskip
For networks (5) and (6) in Fig.~\ref{Figure 2}, the computation was interrupted as it took long. In both networks $s=6$, 
and the first four Hurwitz determinants could be computed. 
These determinants are polynomials in $\k$ and $x$ with positive coefficients, thus they are positive for every positive steady state. In order to compute $H_5$, we followed the second strategy explained in Section~\ref{sec:comp_challenges}. We rely on an identity that holds for the Hurwitz determinants of a generic polynomial of degree 6. Namely, for a generic polynomial $h(t)=a_6t^6+a_5t^5+a_4t^4+a_3t^3+a_2t^2+a_1t+a_0$, the fifth Hurwitz determinant can be written in terms of the previous ones as follows \[H_5=a_1H_4+a_0 (-a_0a_5^3+a_1a_5H_2-a_3H_3).\] 
With this identity, we analyze the sign of the coefficients of $H_5$ by studying separately the coefficients of $A=a_0(-a_0a_5^3+a_1a_5H_2-a_3H_3)$ and $B=a_1H_4$ after substituting $a_i$ for the coefficient of $\lambda^i$ in $q_x(\lambda)$. Note that the coefficients of $B$ are positive because both $a_1$ and $H_4$ are polynomials with positive coefficients. 
We proceed to find the coefficients of $H_5$ by selecting the terms of $A$ that might be negative. We actually found that all coefficients of $A$ are polynomials in $\k$ with some negative coefficients. 
When matched with monomials in $B$, all negative coefficients cancel out, confirming that $H_5$ is positive for all positive $\k$ and $x$.
Therefore, $H_5$ is positive for every positive steady state. 
Knowing this about $H_5$, we also conclude that $H_6=a_0H_5$ is a polynomial with positive coefficients and the unique steady state in each stoichiometric compatibility class is exponentially stable for every set of parameters. 
Note that in these computations we do not need to evaluate at a positive parametrization, meaning all Hurwitz determinants are positive for arbitrary positive $\k$ and $x$. 

In networks (5) and (6), this strategy to compute $H_5$ meant that we were analyzing only 24,196 and 27,982 coefficients instead of 37,319 and 36,970 coefficients respectively.

\subsection{Multistationary networks from Fig.~\ref{Figure 3}}
We consider now the networks in Fig.~\ref{Figure 3}, which all are known to be multistationary.
We sketch here why the procedure fails for each network, and how it applies to the reduced network.  
To certify multistationarity for the reduced networks, we  will apply the method from \cite{FeliuPlos}, which consists of finding values for the rate constants and concentration variables such that $\det(J_{F_T}(\phi))$ is negative, where $\phi$ is a parametrization of the steady states.

\medskip
\noindent
\textbf{Network (a)}. This network is the combination of two one-site modification cycles where the same kinase \ce{E} activates the phosphorylation process and two different phosphates \ce{F1} and \ce{F2} catalyze the dephosphorylation process:
	\begin{align*}
		&\ce{S0 + E <=>[\k_1][\k_2] S0 E ->[\k_3]S1 +E} & 
		&\ce{S1 + E <=>[\k_4][\k_5] S1 E ->[\k_6] S2 + E} \\
		&\ce{S1 + F1 <=>[\k_7][\k_8] S1 F1 ->[\k_9] S0 + F1} &
		&\ce{S2 + F2 <=>[\k_{10}][\k_{11}] S2 F2 ->[\k_{12}] S1 + F2.} 
	\end{align*}		
The species are renamed as $S_0 = X_1, S_1 = X_2, S_2 = X_3, E = X_4, F_1 = X_5, F_2 = X_6, S_0E = X_7, S_1E = X_8, S_1F_1 = X_9, S_2F_2 = X_{10}$. 
 Since the polynomial $q_x$ has degree 6,  we need to compute 6 Hurwitz determinants. These determinants were computed and their signs were analyzed up to $H_4$, and they have positive coefficients. However, the analysis of the sign of $H_5$ was interrupted as it took long and it was not possible to store the polynomial in the expanded format in a regular PC. To compute and study this determinant more effectively, we use a monomial positive parametrization $\phi$ of the steady state variety, which, in the notation  given in
Section~\ref{sec:comp_challenges}, corresponds to 
{\small
\begin{align*}
	\beta= & \left( \frac{(\k_2+\k_3)(\k_5+\k_6)\k_7\k_9\k_{10}\k_{12}}{\k_1\k_3\k_4\k_6(\k_8+\k_9)(\k_{11}+\k_{12})}, \frac{(\k_5+\k_6)\k_{10}\k_{12}}{\k_4\k_6(\k_{11}+\k_{12})}\right.,1,1,1,1,
\frac{(\k_5+\k_6)\k_7\k_9\k_{10}\k_{12}}{\k_3\k_4\k_6(\k_8+\k_9)(\k_{11}+\k_{12})}, \\ &\qquad  \frac{\k_{10}\k_{12}}{\k_6(\k_{11}+\k_{12})},  \frac{(\k_5+\k_6)\k_7\k_{10}\k_{12}}{\k_4\k_6(\k_8+\k_9)(\k_{11}+\k_{12})},\left. \frac{ \k_{10}}{\k_{11}+\k_{12}}\right),
\end{align*}
\[B=\left[\begin{array}{cccccccccc}
	1&1&1&0&0&0&1&1&1&1 \\
	-2&-1&0&1&0&0&-1&0&-1&0\\
	1&0&0&0&1&0&1&0&1&0\\
	1&1&0&0&0&1&1&1&1&1
\end{array}\right]  \qquad \mbox{ \normalsize{ and} } \qquad \xi=(x_3,x_4,x_5,x_6).\] }%
Using   the identification of the monomials with tuples, it was possible to compute $H_5(\phi)$. However, the sign of this function remains unclear since it has coefficients with different signs. 

\smallskip
We consider next the reduced network obtained by first removing all the reverse reactions and then the intermediates \ce{S1 F1} and \ce{S2 F2}. When removing these intermediates the reactions \ce{S1 + F1 -> S1 F1 -> S0 + F1} and \ce{S2 + F2 -> S2 F2 -> S1 + F2 } become \ce{S1 + F1 -> S0 + F1} and \ce{S2 + F2 -> S1 + F2 } respectively. The reduced network is
\begin{align*}
	&\ce{S0 + E ->[\t_1] S0 E ->[\t_2]S1 +E} 
	&\ce{S1 + E ->[\t_3] S1 E ->[\t_4] S2 + E} \\
	&\ce{S1 + F1 ->[\t_5] S0 + F1} 
	&\ce{S2 + F2 ->[\t_6] S1 + F2}.
\end{align*}	
The species are now renamed as $S_0= X_1, S_1 = X_2, S_2 = X_3, E = X_4, F_1 = X_5, F_2 = X_6, S_0E = X_7, S_1E = X_8$. The polynomial $q_x$ associated with this network has degree 4 and, when computing the Hurwitz determinants we have that $H_1,H_2$ and $H_3$ are positive. However the sign of $H_4$ is unclear even after evaluating at a positive parametrization of the steady state variety $\phi$. In this situation we explore the possibility of applying Theorem~\ref{thm:bistability} to deduce bistability. 
The conservation laws of the system are 
\[x_1+x_2+x_3+x_7+x_8=T_1, \quad x_4+x_7+x_8=T_2, \quad x_5=T_3 \quad \mbox{ and } x_6=T_4.\]
Taking the indices $i_1,i_2,i_3,i_4$ as $1,4,5,6$ respectively, we construct $F_T$ as in Eq.~\eqref{F}. The solutions of $F_{T,\ell}=0$ for $\ell\neq 1$ are written in terms of $z=x_2$ as 
{\footnotesize \[\varphi (z) = \left(\frac{\t_2\t_5 T_3(\t_3z +\t_4)z}{\t_1\t_4(   \t_2 T_2 - \t_5 T_3 z)},z,\frac{\t_3\t_4z(   \t_2 T_2 - \t_5 T_3 z)}{\t_2\t_6 T_4 (\t_3z+\t_4)}, \frac{\t_4(   \t_2 T_2 - \t_5 T_3 z)}{\t_2(\t_3z+\t_4)}, T_3, T_4, \frac{\t_5 T_3z}{\t_2}, \frac{\t_3(   \t_2 T_2 - \t_5 T_3 z)z }{\t_2(\t_3 z +\t_4)}\right),\]}%
for  $z\in\mathcal{E}$, where $\mathcal{E}=\left(0,\frac{T_2\t_2}{T_3\t_5} \right)$. 

Note that $\varphi_2(z)=z$ and $\varphi_2'(z)=1\neq 0$. This means that  the positive steady states in the stoichiometric compatibility class defined by $T$ are in one to one correspondence with the positive roots of $F_{T,1}(\varphi(z))$ in $\mathcal{E}$. This rational function, presented below, has as numerator a polynomial of degree 3.
{\footnotesize 
\begin{equation*}
\begin{split}
	 F_{T,1}(\varphi(z))=&\frac{1}{T_4\t_1\t_2\t_4\t_6(T_3\t_5z - T_2\t_2)(\t_3z + \t_4)}\left( -T_3\t_3\t_5(T_3\t_1\t_4^2\t_5 - T_4\t_1\t_2\t_4\t_6 + T_4\t_2^2\t_3\t_6)z^3 - \right. \\
	 & \t_4(T_1T_3T_4\t_1\t_2\t_3\t_5\t_6 - T_2T_3T_4\t_1\t_2\t_3\t_5\t_6 - T_3^2T_4\t_1\t_4\t_5^2\t_6 - 2T_2T_3\t_1\t_2\t_3\t_4\t_5 + T_2T_4\t_1\t_2^2\t_3\t_6 
 - \\ & T_3T_4\t_1\t_2\t_4\t_5\t_6  + 2T_3T_4\t_2^2\t_3\t_5\t_6)z^2 +  \t_2\t_4(T_1T_2T_4\t_1\t_2\t_3\t_6 - T_1T_3T_4\t_1\t_4\t_5\t_6 - T_2^2T_4\t_1\t_2\t_3\t_6 \\ &- T_2T_3T_4\t_1\t_4\t_5\t_6   - T_2^2\t_1\t_2\t_3\t_4 - T_2T_4\t_1\t_2\t_4\t_6 -  T_3T_4\t_2\t_4\t_5\t_6)z 
	  + \left. T_1T_2T_4\t_1\t_2^2\t_4^2\t_6 \right).
\end{split}
\end{equation*} }%
We have already shown that the second hypothesis of Theorem~\ref{thm:bistability} holds. For the first hypothesis, a straightforward computation shows that
\[\det(J_{F_T}(\varphi(z))_{J,I})=\t_1\t_4\t_6T_4(\t_5 T_3 z - \t_2 T_2),\]
which is negative for every $z\in\mathcal{E}$. 
We further have $s=4$, $i=2$, $j=1$, and the independent term of the numerator of $F_{T,1}(\varphi(z))$ is negative, meaning that   $(F_{T,j}\circ \varphi)'(z_1)>0$. This gives that the sign of Eq.~\eqref{eq:importantsign}
is $(-1)^{4+1+2} (+1)(-1)=1$ positive. 
Using Theorem~\ref{thm:bistability}, we conclude that for every set of parameters such that $F_{T,1}(\varphi(z))$ has three roots $z_1<z_2<z_3$ in $\mathcal{E}$, the steady states $\varphi(z_1),\varphi(z_3)$ are exponentially stable and $\varphi(z_2)$ is unstable. 

\smallskip
All that is left is to show that the reduced network admits three positive steady states in some stoichiometric compatibility class for some choice of $\t$, or what is the same, that $F_{T,1}(\varphi(z))$ admits three roots in $\mathcal{E}$. 
We consider $\det(J_{F_T}(\phi))$ for a parametrization $\phi$ of the steady states:
{\footnotesize
\begin{align*}
\phi(x_3,x_4,x_5,x_6) &=\left( \frac{\t_5\t_6x_3x_5x_6}{\t_1\t_3x_4^2}, \frac{\t_6x_3x_6}{\t_3x_4},x_3 ,x_4x,x_5,x_6, \frac{\t_5\t_6x_3x_5x_6}{\t_2\t_3x_4}, \frac{\t_6x_3x_6}{\t_4}\right) \qquad \textrm{\normalsize{ and }}\\
\det(J_{F_T}(\phi)) & =\tau_{1}\tau_{2}\tau_{6}^{2}x_{3}x_{6}^{2}-\tau_{1}\tau_{4}\tau_{5}\tau_{6} x_{3}x_{5}x_{6}
+2\frac{\tau_{2}\tau_{5}\tau_{6}^{2}x_{3}x_{5}x_{6}^{2}}{x_{4}}+
\frac{\tau_{4}\tau_{5}^{2}\tau_{6}^{2}x_3x_{5}^{2}x_{6}^{2}}{\tau_{3}x_{4}^{2}}\\ & \quad +
\tau_{1}\tau_{2}\tau_{3}\tau_{6}x_{4}^{2}x_{6}+\tau_{1}\tau_{4}\tau_{5}\tau_{6}x_{4}x_{5}x_{6} +\tau_{1}\tau_{2}\tau_{3}\tau_{4}x_{4}^{2}+\tau_{1}\tau_{2}\tau_{4}\tau_{6}x_{4}x_{6}+\tau_{2}\tau_{4}\tau_{5}\tau_{6}x_{5}x_{6}.
\end{align*} }%
By letting $\tau_i=1$, for $i=1,\dots,6$ and $x_3=100, x_4=10,x_5=10,x_6=1$, $\det(J_{F_T}(\phi))=-280$, which is negative. Hence  applying the method in \cite{FeliuPlos}, we conclude that the stoichiometric compatibility class containing $\phi(100,10,10,1)$ has more than one positive steady state. Specifically, this class corresponds to $T_1=320, T_2=210, T_3=10,T_4=1$. Either by solving the steady state equations or finding the roots of $F_{T,1}(\varphi(z))$  for this choice of parameters, we confirm that the system has three positive steady states. 

Therefore, the reduced network is bistable for all choice of parameter values for which there are three positive steady states, and the original network admits bistability in some region of the parameter space.

\medskip
\noindent
\textbf{Network (b)}. This network is the combination of two one-site modification cycles in a cascade, where the same phosphatase \ce{F} acts in both layers.
\begin{align*}
	&\ce{S0 + E <=>[\k_1][\k_2] S0 E ->[\k_3]S1 +E} 	&\ce{S1 + F <=>[\k_4][\k_5] S1 F ->[\k_6] S0 + F}\\ 
	&\ce{P0 + S1 <=>[\k_7][\k_8] P0 S1 ->[\k_9] P1 + S1} 	&\ce{P1 + F <=>[\k_{10}][\k_{11}] P1 F ->[\k_{12}] P0 + F}.
\end{align*}

We rename the species as $E = X_1, F = X_2, S_0 = X_3, S_1 = X_4, P_0 = X_5, P_1 = X_6, S_0E = X_7, S_1F= X_8, P_0S_1 = X_9, P_1F= X_{10}$. For this network the polynomial $q_x$ has degree 6; and after some computations it was possible to prove that $H_1,H_2,H_3$ are positive polynomials. However, the sign of $H_4$ is unclear and the direct computation of $H_4(\phi)$ was not feasible as the memory in a regular PC was not enough. The positive steady state variety has a monomial parametrization $\phi$. We use the identification of monomials with tuples, to compute and analyze $H_4(\phi)$ more efficiently. With the notation in Section~\ref{sec:comp_challenges}, $\phi$ corresponds to
{\footnotesize \begin{align*}
	\beta & =\left( \frac{\k_4\k_6\k_{10}\k_{12}(\k_2+\k_3)(\k_8+\k_9)}{\k_1\k_3\k_7\k_9(\k_5\k_{11}+\k_5\k_{12}+\k_6\k_{11}+\k_6\k_{12})}\right.,  1, 1, \frac{(\k_8+\k_9)\k_{10}\k_{12}}{(\k_{11}+\k_{12})\k_7\k_9}, 1, 1,
\\
	& \frac{ \k_4\k_6\k_{10}\k_{12}(\k_8+\k_9)}{\k_3\k_7\k_9(\k_5\k_{11}+\k_5\k_{12}+\k_6\k_{11}+\k_6\k_{12})},
\frac{\k_4\k_{10}\k_{12}(\k_8+\k_9)}{\k_7\k_9(\k_5\k_{11}+\k_5\k_{12}+\k_6\k_{11}+\k_6\k_{12})}, \frac{\k_{10}\k_{12}}{\k_9(\k_{11}+\k_{12})}, \left. \frac{\k_{10}}{\k_{11}+\k_{12}}\right),
\end{align*}  
\[B= \left[ \begin{array}{cccccccccc}
	2&1&0&2&0&0&2&2&1&1\\
	-1&0&1&0&0&0&0&0&0&0\\
	-1&0&0&-1&1&0&1&-1&0&0\\
	1&0&0&1&0&1&1&1&1&1
\end{array}  \right] \qquad \mbox{\normalsize{ and }} \qquad \xi=(x_2,x_3 ,x_5,x_6).\]}%

Using this identification it was possible to compute $H_4(\phi)$. However, its sign was still unclear as we encountered both positive and negative coefficients. 

\smallskip
We then proceeded to reduce the network by  removing all reverse reactions and the intermediates \ce{S0 E} and \ce{S1 F}. That is, the reactions \ce{S0 + E -> S0 E ->S1 +E} and \ce{S1 + F -> S1 F -> S0 + F} are transformed into \ce{S0 + E ->S1 +E} and \ce{S1 + F ->[\t_2] S0 + F} respectively. We are left with the following reduced network 
\begin{align*}
	&\ce{S0 + E ->[\t_1]S1 +E} 	&\ce{S1 + F ->[\t_2] S0 + F}\\
	&\ce{P0 + S1 ->[\t_3] P0 S1 ->[\t_4] P1 + S1} 
	&\ce{P1 + F ->[\t_5] P1 F ->[\t_6] P0 + F}.
\end{align*}
In this network we rename the species as $E = X_1, F = X_2, S_0 = X_3, S_1 = X_4, P_0 = X_5, P_1 = X_6, P_0S_1 = X_9, P_1 F = X_{10}$. The polynomial $q_x$ has degree 4 and, after computing the Hurwitz determinants, we have that $H_1,H_2,H_3$ are positive. Therefore, the stability of the positive steady states depends on the sign of $H_4(\phi)$.  

The conservation laws of the system are
\[x_1=T_1, \quad x_5+x_6+x_9+x_{10}=T_2,\quad x_3+x_4+x_9=T_3,\quad x_2+x_{10}=T_4.\]
Taking the indices $i_1=1,i_2=2,i_3=3,i_4=5$ we define $F_T$ as in Eq.~\eqref{F}. Furthermore, the solutions of $F_{T,\ell}=0$ for $\ell\neq 5$ can be positively parametrized in terms of $z=x_6$ as 
{\footnotesize
\begin{align*}
	\varphi(z)=\left(T_1, \frac{\t_6T_4}{\t_5z+\t_6},\right. 
\frac{T_4\t_2\t_6( (T_3\t_4-T_4\t_6)\t_5z+T_3\t_4\t_6)}{\t_4(T_1\t_1\t_5z+T_1\t_1\t_6+T_4\t_2\t_6)(\t_5z+\t_6)}, \frac{T_1\t_1( (T_3\t_4-T_4\t_6)\t_5z+T_3\t_4\t_6)}{\t_4(T_1\t_1\t_5z+T_1\t_1\t_6+T_4\t_2\t_6)},\\
	\frac{\t_5\t_6T_4\t_4(T_1\t_1\t_5z+T_1\t_1\t_6+T_4\t_2\t_6)z}{T_1\t_1\t_3( (T_3\t_4-T_4\t_6)\t_5z+T_3\t_4\t_6)(\t_5z+\t_6)},z,\left.\frac{\t_5\t_6T_4z}{\t_4(\t_5z+\t_6)},\frac{T_4\t_5z}{\t_5z+\t_6}\right)
\end{align*} }%
for  $z\in \mathcal{E}$, where $\mathcal{E}=\R_{>0}$ if $T_3\t_4-T_4\t_6>0$ or $\mathcal{E}=\left( 0,\frac{T_3\t_4\t_6}{\t_5(T_4\t_6-T_3\t_4)}\right)$ if $T_3\t_4-T_4\t_6\leq 0$.  Note that the positive steady states in the stoichiometric compatibility class defined by $T$ are in one to one correspondence with the zeros of $F_{T,5}(\varphi(z))$, below, contained in $\mathcal{E}$. 
{\footnotesize 
\begin{equation*}
\begin{split}
	F_{T,5}(\varphi(z))=&\frac{1}{T_1\t_1\t_3\t_4(T_3\t_4\t_5z - T_4\t_5\t_6z + T_3\t_4\t_6)(\t_5z +  \t_6)} \left[ T_1\t_1\t_3\t_4\t_5^2(T_3\t_4 - T_4\t_6)z^3 \right. - T_1\t_1\t_5 (T_2T_3\t_3\t_4^2\t_5 \\ &
- T_2T_4\t_3 \t_4\t_5\t_6  - T_3T_4\t_3\t_4^2\t_5 -  T_3T_4\t_3\t_4\t_5\t_6+ T_4^2\t_3\t_4\t_5\t_6 + T_4^2\t_3\t_5\t_6^2 - 2T_3\t_3\t_4^2\t_6 + T_4\t_3\t_4\t_6^2\\ &  - T_4\t_4^2\t_5\t_6)z^2  - \t_4\t_6(2T_1 T_2T_3\t_1\t_3\t_4\t_5  - T_1T_2T_4\t_1\t_3\t_5\t_6 -  T_1T_3T_4\t_1\t_3\t_4\t_5 \\
	   &- T_1T_3T_4\t_1\t_3\t_5\t_6 - T_1T_3\t_1\t_3\t_4\t_6 - T_1T_4\t_1\t_4\t_5\t_6 - T_4^2\t_2\t_4\t_5\t_6)z - T_1T_2T_3\t_1\t_3\t_4^2\t_6^2 ].
\end{split}
\end{equation*}  }%

The numerator of this univariate rational function has degree 3 and the denominator is positive in $\mathcal{E}$. Additionally, $\varphi_6(z)=z$ and $\varphi_6'(z)=1\neq 0$. With this parametrization where $i=6,j=5$, we have all the elements required in the statement of Theorem~\ref{thm:bistability}. We also know from the analysis of the Hurwitz determinants that the second hypothesis of Theorem~\ref{thm:bistability} holds. The first hypothesis of the theorem also holds as 
\[\det(J_{F_T}(\varphi(z))_{J,I})=-\t_3T_1\t_1((T_3\t_4-T_4\t_6)\t_5z+T_3\t_4\t_6)\]
is negative for every $z\in\mathcal{E}$. Using Theorem~\ref{thm:bistability}, and the fact that the independent term of $F_{T,5}(\varphi(z))$ is negative, we conclude that for every set of parameters such that $F_{T,5}(\varphi(z))$ has 3 positive roots $z_1<z_2<z_3$ in $\mathcal{E}$, the sign of  Eq.~\eqref{eq:importantsign}
 is $(-1)^{6+5+4}(1)(-1)=1$   and thus $\varphi(z_1),\varphi(z_3)$ are exponentially stable and $\varphi(z_2)$ is unstable. 

\smallskip
It remains to see that there is a set of parameters such that the network admits three positive steady states. We find $\det(J_{F_T}(\phi))$  for a parametrization $\phi$:
{\footnotesize \begin{align*}
\phi(x_2,x_3,x_5,x_6)&=  \left( \frac{\t_2\t_5x_2^2x_6}{\t_1\t_3x_3x_5},x_2,x_3,\frac{\t_5x_2x_6}{\t_3x_5},x_6 ,\frac{\t_5x_2x_6}{\t_4}, \frac{\t_5x_2x_6}{\t_6}  \right) \qquad \textrm{ and } \\
\det(J_{F_T}(\phi)) & =\frac{1}{\t_3x_3x_5^2} \Big(\t_2\t_5x_2^2(\t_3\t_4\t_5x_2x_3x_5x_6-\t_3\t_4\t_5x_3x_5^2x_6+\t_3\t_4\t_5x_3x_5x_6^2+\t_3\t_5\t_6x_2x_3x_5x_6\\
& \qquad +\t_3\t_5\t_6x_2x_5^2x_6  +\t_4\t_5^2x_2^2x_6^2+\t_4\t_5^2x_2x_6^3+\t_5^2\t_6x_2^2x_6^2+\t_3\t_4\t_6x_3x_5^2+\t_3\t_4\t_6x_3x_5x_6\\
& \qquad+\t_4\t_5\t_6x_2x_5x_6+\t_4\t_5\t_6x_2x_6^2)\Big).
\end{align*}}%
Taking $\t_i=1$ for $i=1,\ldots ,6$ and $x_2 = 1, x_3 = 10, x_5 = 20, x_6 = 10$ this determinant is -4500. By \cite{FeliuPlos}, the stoichiometric compatibility class containing $\phi(1,10,20,10)$ has more than one positive steady state. The total amounts defining it are $T_1 = \frac{1}{20}, T_2 = 50, T_3 = \frac{41}{2}, T_4 = 11$. Using these parameters and solving either $F_T(x)=0$ or finding the roots of $F_{T,5}(\varphi)$ we verify that this stoichiometric compatibility class has in fact three positive steady states. 

We conclude that the reduced network is bistable for every set of parameters for which there are three positive steady states, and the original network admits bistability in some region of the parameter space.

\medskip
\noindent
\textbf{Network (c)}. This network is the combination of two one-site modification cycles in a cascade, where the same kinase \ce{E} acts in both layers:
\begin{align*}
	&\ce{S0 + E <=>[\k_1][\k_2] S0 E ->[\k_3] S1 + E} & 
	&\ce{S1 + F1 <=>[\k_4][\k_5] S1 F1 ->[\k_6] S0 + F1}\\ 
	&\ce{P0 + S1 <=>[\k_7][\k_8] P0 S1 ->[\k_9] P1 + S1} &
	&\ce{P0 + E <=>[\k_{13}][\k_{14}] P0 E ->[\k_{15}] P1 + E}\\
	&\ce{P1 + F2 <=>[\k_{10}][\k_{11}] P1 F2 ->[\k_{12}] P0 + F2}.
\end{align*}
The species are renamed as $S_0 = X_1, S_1 = X_2, P_0 = X_3, P_1 = X_4, E = X_5, F_1 = X_6, F_2 = X_7, S_0E = X_8, S_1F_1 = X_9, P_0S_1 = X_{10}, P_1F_2= X_{11}, P_{0}E = X_{12}$. For this network, the polynomial $q_x$ has degree 7 and the determinants $H_1,H_2,H_3$ are positive polynomials. However, the computations of the other determinants was not possible as there was not enough memory to store the computations in a regular PC. In this case, we could not parametrize the positive steady state variety by monomials; therefore, it is not possible to use the identification between monomials and tuples to analyze the sign of the remaining determinants.

\medskip
We proceed directly to reduce the network by removing all the reverse reactions first, and then the intermediates \ce{S0 E},\ce{P0 S1} and \ce{P1 F2}. That is, the reactions \ce{S0 + E -> S0 E -> S1 + E}, \ce{P0 + S1 -> P0 S1 -> P1 + S1} and \ce{P1 + F2 -> P1 F2 -> P0 + F2} are transformed into \ce{S0 + E  -> S1 + E}, \ce{P0 + S1  -> P1 + S1} and \ce{P1 + F2  -> P0 + F2} respectively. We are left with the following network
\begin{align*}
	&\ce{S0 + E  ->[\t_1] S1 + E} & 
	&\ce{S1 + F1  ->[\t_2] S1 F1 ->[\t_3] S0 + F1}\\ 
	&\ce{P0 + S1  ->[\t_4] P1 + S1}&
	&\ce{P0 + E ->[\t_6] P0 E ->[\t_7] P1 + E} & 
	&\ce{P1 + F2  ->[\t_5] P0 + F2}.
\end{align*}
The species are renamed as $S_0 = X_1, S_1 = X_2, P_0 = X_3, P_1 = X_4, E = X_5, F_1 = X_6, F_2 = X_7, S_1F_1 = X_9, P_0E = X_{12}$. The polynomial $q_x$ associated with this network has degree 4 and $H_1,H_2,H_3$ are positive after evaluating in the following positive parametrization $\phi$ of the steady state variety:

{\small
\[\phi(x_1,x_2,x_3,x_6,x_7)=\left(x_1,x_2,x_3,\frac{x_2x_3(\t_1\t_4x_1 + \t_2\t_6x_6)}{\t_1\t_5x_1x_7}, \frac{\t_2x_2x_6}{\t_1x_1},x_6,x_7, \frac{\t_2x_2x_6}{\t_3}, \frac{ \t_6x_3\t_2x_2x_6}{\t_1x_1\t_7}  \right).\]}%
However, the sign of $H_4$ is unclear. The conservation laws of the system are 
\[x_1+x_2+x_9=T_1,\quad x_3+x_4+x_{12}=T_2,\quad x_5+x_{12}=T_3, \quad x_6+x_9=T_4,\quad x_7=T_5.\]
Taking $i_1=1,i_2=3,i_3=5,i_4=6,i_5=7$ we define $F_T$ as in Eq.~\eqref{F}. Additionally, the solutions of $F_{T,\ell}(x)=0$ for $\ell\neq 6$ can be   parametrized in terms of $z=x_3$ as 
{\footnotesize
\begin{align*}
\varphi(z)=&\left(\frac{b_1(z) \t_3}{(\t_3\t_6z+T_3\t_1\t_7+\t_3\t_7)\t_4z}, 
\frac{b_2(z)}{\t_4z(\t_6z+\t_7)},z,\frac{(-\t_6z^2+(T_2 \t_6-T_3\t_6-\t_7)z+T_2\t_7)}{\t_6z+\t_7}, \frac{T_3\t_7}{\t_6z+\t_7},\right. \\ &  \quad 
\left.\frac{b_1(z) T_3\t_1\t_3\t_7}{\t_2(\t_3\t_6z+T_3\t_1\t_7+\t_3\t_7)b_2(z)},T_5,
\frac{b_1(z) T_3\t_1\t_7}{\t_4z(\t_6z+\t_7) (\t_3\t_6z+T_3\t_1\t_7+\t_3\t_7)}, \frac{\t_6T_3z}{\t_6z+\t_7}\right),
\end{align*}}%
where

{\small
\begin{align*}
b_1(z) & =(T_1\t_4+T_5\t_5)\t_6z^2 + (-T_2T_5\t_5\t_6+T_3T_5\t_5\t_6+T_1\t_4\t_7+T_3\t_6 \t_7+T_5\t_5\t_7)z-T_2T_5\t_5\t_7, \\
b_2(z) & =-T_5\t_5\t_6z^2+(T_2T_5 \t_5\t_6-T_3T_5 \t_5\t_6-T_3\t_6\t_7-T_5\t_5\t_7)z+T_2T_5\t_5\t_7. 
\end{align*}}%
The parametrization is positive if and only if  $b_1(z),b_2(z)>0$. This happens for $z\in \mathcal{E}$, where $\mathcal{E}=(\beta_1,\beta_2)$ with $\beta_1$ and $\beta_2$ the (only) positive roots of the polynomials $b_1$ and $b_2$ respectively.
 
The steady states in each stoichiometric compatibility class are in one to one correspondence with the roots of $F_{T,6}(\varphi(z))$   in $\mathcal{E}$. The numerator of this function is a polynomial of degree 5 and the denominator is positive in $\mathcal{E}$. Additionally, $\varphi_3(z)=z$ and $\varphi_3'(z)=1\neq 0$. With this parametrization and taking $i=3,j=6,s=4$, we have all the elements in the statement of Theorem~\ref{thm:bistability}. The first three Hurwitz determinants are positive and thus, the second hypothesis of the theorem holds. The first hypothesis  holds as
\begin{align*}
\det(J_F(\varphi(z))_{J,I})=\frac{-\t_2(\t_3\t_6z+T_3\t_1\t_7+\t_3\t_7)\, b_2(z)}{\t_6z+\t_7},
\end{align*}
 is negative for all $z\in\mathcal{E}$. 
We need to decide the sign of $(F_{T,j}\circ \varphi)'(z_1)$ at the first root of the numerator of $F_{T,j}\circ \varphi$. Indirect evaluation of  $F_{T,j}\circ \varphi$ at $\beta_1$, by isolating $T_2$ from $b_2(z)=0$ and substitution into  $F_{T,j}\circ \varphi$, shows that  $F_{T,j}\circ \varphi$ is negative at $\beta_1$. Hence the derivative at $z_1$ is positive. 
By Theorem~\ref{thm:bistability}, we conclude that for every set of parameters such that $F_{T,6}(\varphi(z))$ has more than two positive roots $z_1,z_2,\ldots$, the sign of Eq.~\eqref{eq:importantsign}  
 is $(-1)^{3+6+4}(1)(-1)=1$ positive. Therefore, the steady states $\varphi(z_1),\varphi(z_3),\ldots$ are exponentially stable and $\varphi(z_2),\ldots$ are unstable. 

It remains to see that there is a set of parameters such that the stoichiometric compatibility class in fact contains more than two positive steady states.  We find $\det(J_F(\phi))$, with $\phi$ as above:
{\small
\begin{align*}
	\det(J_F(\phi))=&-\t_1\t_2^2\t_4\t_6x_1^2x_2^2x_3x_6+\t_1\t_2\t_3\t_5\t_6x_1^2x_3x_6x_7+\t_1\t_2\t_3\t_4\t_7x_1^2x_2x_6+\t_1\t_2\t_3\t_5\t_7x_1^2x_6x_7  \\
	& +\t_1\t_2^2\t_4\t_6x_1x_2^3x_3x_6 + \t_1\t_2^2\t_4\t_7x_1x_2^3x_6+\t_1\t_2^2\t_4\t_6x_1x_2^2x_3x_6^2+\t_1\t_2^2\t_5\t_6x_1x_2^2x_3x_6x_7\\
	& +\t_1\t_2^2\t_4\t_7x_1x_2^2x_6^2+\t_1\t_2^2\t_5\t_7x_1x_2^2x_6x_7 + \t_1\t_2\t_3\t_4\t_6x_1x_2^2x_3x_6 +\t_1\t_2\t_3\t_4\t_7x_1x_2^2x_6 \\
&+\t_1\t_2^2\t_5\t_6x_1x_2x_3x_6^2x_7+\t_1\t_2\t_3\t_5\t_6x_1x_2x_3x_6x_7 +\t_2^2\t_5(\t_1\t_7+\t_3\t_6)x_1x_2x_7x_6^2\\
&+ \t_1\t_2\t_3\t_5\t_7x_1x_2x_6x_7+\t_2^2\t_3\t_6\t_7x_1x_2x_6^2+\t_2^3\t_5\t_6x_2^3x_6^2x_7+x_2^2\t_5x_7\t_2^3\t_6x_6^3 \\
	&+\t_2^3\t_6\t_7x_2^3x_6^2  + \t_2^3\t_6\t_7x_2^2x_6^3+\t_2^2\t_3\t_5\t_6x_2^2x_6^2x_7+\t_2^2\t_3\t_6\t_7x_2^2x_6^2.
\end{align*} }%
Taking $\t_i=1$ for $i=1,\ldots ,7$ and $x_1 = 40, x_2 = 10, x_3 = 1, x_6 = 1, x_7 = 1$, this determinant is $-32,000$. By  \cite{FeliuPlos},  it follows that  the stoichiometric compatibility class containing $\phi(x)$ has more than one positive steady state. The total amounts defining the stoichiometric compatibility class are $T_1= 60, T_2 = \frac{23}{2}, T_3 = \frac{1}{2}, T_4 = 11, T_5 = 1$. Using these parameters and solving either $F_T(x)=0$ or $F_{T,6}(\varphi(z))=0$ we verify that this stoichiometric compatibility class has three positive steady states as desired. 

The reduced network is thus bistable whenever it has  three positive steady states, and the original network admits bistability in some   parameter region.

\medskip
\textbf{Network (d)}. In this network a kinase \ce{E} and a phosphatase \ce{F} act on a substrate \ce{S0} and the two sites of its phosphorylation \ce{S1} and \ce{S2}.
\begin{align*}
	&\ce{E + S0 <=>[\k_1][\k_2] S0 E ->[\k_3] E + S1 <=>[\k_7][\k_8] S1 E ->[\k_9]S2 + E }\\
	&\ce{F + S2 <=>[\k_{10}][\k_{11}] S2 F ->[\k_{12}] F + S1 <=>[\k_4][\k_5] S1 F ->[\k_6] F + S0}
\end{align*}
The species are renamed as $E = X_1, F = X_2, S_0 = X_3, S_1 = X_4, S_2 = X_5, S_0E = X_6, S_1E = X_7, S_2F = X_8, S_1F = X_9$. In this network the polynomial $q_x$ has degree 6 and, the first three Hurwitz determinants are polynomials with positive coefficients; therefore, they are positive for every positive steady state. Regarding the determinants $H_4$ and $H_5$, the sign of $H_4$ is unclear and the analysis of the sign of $H_5$ was stopped as the computations could not be stored in a regular PC. In this case, the positive steady state variety can be parametrized by monomials $\phi$, which,  with the notation in Section~\ref{sec:comp_challenges}, corresponds to 
{\footnotesize \begin{align*}
	\beta=&\left( 1,1, \frac{(\k_2+\k_3)\k_4\k_6\k_{10}\k_{12}(\k_8+\k_9)}{\k_1\k_3\k_7\k_9(\k_5+\k_6)(\k_{11}+\k_{12})} , \frac{(\k_8+\k_9)\k_{10}\k_{12}}{\k_7\k_9(\k_{11}+\k_{12})},1,\frac{ \k_4\k_6\k_{10}\k_{12}(\k_8+\k_9)}{\k_3\k_7\k_9(\k_5+\k_6)(\k_{11}+\k_{12})},  \right. \\
	& \left.\frac{\k_{10}\k_{12}}{\k_9(\k_{11}+\k_{12})}, \frac{\k_{10}}{\k_{11}+\k_{12}},\frac{\k_4\k_{10}\k_{12}(\k_8+\k_9)}{\k_7\k_9(\k_5\k_{11}+\k_5\k_{12}+\k_6\k_{11}+\k_6\k_{12})} \right),
\end{align*} 
\[B=\left[\begin{array}{ccccccccc}
1 & 0 & -2 & -1 & 0 & -1 & 0 & 0 & -1\\
0 & 1 & 2 & 1 & 0 & 2 & 1 & 1 & 2 \\
0 & 0 & 1 & 1 & 1 & 1 & 1 & 1 & 1
\end{array}  \right] \qquad \mbox{ and } \qquad \xi=(x_1,x_2,x_5).\]}%
Using  tuples, it was possible to compute $H_4(\phi)$ and $H_5(\phi)$ and, after studying their sign, we found that $H_4(\phi)$ is a positive polynomial, but the sign of $H_5(\phi)$ was not clear. 

We then proceed to reduce the network by removing all the reverse reactions and the intermediates \ce{S_1E}, \ce{S_2F} and \ce{S_1F}. That is, the reactions \ce{S1 + E ->S1 E -> S2 + E }, \ce{S2 + F -> S2 F -> S1 + F} and \ce{S1 + F -> S1 F -> S0 + F} become \ce{S1 + E -> S2 + E }, \ce{S2 + F -> S1 + F} and \ce{S1 + F  -> S0 + F} respectively. The reduced network is 
\begin{align*}
	&\ce{S0 + E ->[\t_1] E S0 ->[\t_2] S1 + E ->[\t_3]S2 + E }
	&\ce{S2 + F ->[\t_4] S1 + F ->[\t_5] S0 + F}.
\end{align*}
We rename the species as $E = X_1, F = X_2, S_0 = X_3, S_1 = X_4, S_2 = X_5, S_0E = X_6$. The polynomial $q_x$ has degree 3 and, after computing the Hurwitz determinants, it was possible to prove that $H_1$ and $H_2$ are polynomials with positive coefficients and, thus positive for every positive steady state. However, the sign of $H_3$ was unclear even after evaluating at the following positive parametrization $\phi$:
\[\phi(x_1,x_2,x_3)=\left(x_1,x_2,x_3,\frac{\t_1x_1x_3}{\t_5x_2}, \frac{\t_3x_1^2\t_1x_3}{\t_4\t_5x_2^2}, \frac{\t_1x_1x_3}{\t_2} \right).\]
The conservation laws of the system are
$x_1+x_6=T_1$, $x_2=T_2$ and $x_3+x_4+x_5+x_6=T_3$.
Taking $i_1=1,i_2=2,i_3=3$ we define $F_T$ as in Eq.~\eqref{F}. The solutions of $F_{T,\ell}(x)=0$ for $\ell \neq 1$ can be parametrized in terms of $z=x_1$ as
{\footnotesize
\begin{align*}
	\varphi (z) = & \left( z, T_2,  \frac{T_2^2T_3\t_2\t_4\t_5}{\t_1\t_2\t_3z^2+T_2\t_1\t_4(T_2\t_5 +\t_2)z+T_2^2\t_2\t_4\t_5} \right. ,\frac{ T_2T_3\t_1\t_2\t_4z}{\t_1\t_2\t_3z^2+T_2\t_1\t_4(T_2\t_5 +\t_2)z+T_2^2\t_2\t_4\t_5}, \\
	&\left.  \frac{T_3\t_1\t_2\t_3z^2}{\t_1\t_2\t_3z^2+T_2\t_1\t_4(T_2\t_5 +\t_2)z+T_2^2\t_2\t_4\t_5}, \frac{T_2^2T_3\t_1\t_4\t_5z}{\t_1\t_2\t_3z^2+T_2\t_1\t_4(T_2\t_5 +\t_2)z+T_2^2\t_2\t_4\t_5}\right).
\end{align*} }%
This parametrization is positive for every $z\in\mathcal{E}=\R_{>0}$ and the positive steady states in the stoichiometric compatibility class are in one to one correspondence with the positive roots of $F_{T,1}(\varphi(z))$. This is a rational function whose numerator is a polynomial of degree $3$ and positive denominator:
{\scriptsize
\begin{equation*}
\begin{split}
	F_{T,1}(\varphi(z))= \frac{1}{\t_1\t_2\t_3z^2+T_2\t_1\t_4(T_2\t_5 +\t_2)z+T_2^2\t_2\t_4\t_5} & \left[ \t_1\t_2\t_3z^3 + (T_2^2\t_1\t_4\t_5 - T_1\t_1\t_2\t_3 + T_2\t_1\t_2\t_4)z^2 + \right.\\
	& \hspace{-1.5cm} \left. (-T_1T_2^2\t_1\t_4\t_5 + T_2^2T_3\t_1\t_4\t_5 - T_1T_2\t_1\t_2\t_4 + T_2^2\t_2\t_4\t_5)z - T_1T_2^2\t_2\t_4\t_5 \right] .
\end{split}
\end{equation*}}%
Therefore, there are at most three positive steady states in each stoichiometric compatibility class. Additionally $\varphi_1(z)=z$ and $\varphi_1'(z)=1\neq 0$. With this parametrization, where, $i=j=1$ and $s=3$, we have all the elements in the statement of Theorem~\ref{thm:bistability}.  The second hypothesis of the theorem holds, and the first  hypothesis also does as
\[\det(J_F(\varphi(z))_{J,I})=-T_2^2\t_1\t_4\t_5z-\t_1\t_2\t_3z^2-T_2\t_1\t_2\t_4z-T_2^2\t_2\t_4\t_5,\]
is negative for every $z\in\R_{>0}$. We are in the setting of Theorem~\ref{thm:bistability} and we conclude that for every set of parameters such that $F_{T,1}(\varphi(z))$ has three positive roots $z_1<z_2<z_3$, the sign of 
$(F_{T,j}\circ \varphi)'(z_1)$ is positive as the independent term of $(F_{T,j}\circ \varphi)(z)$ is negative. Furthermore, the sign of  
Eq.~\eqref{eq:importantsign}  
is equal to $(-1)^{1+1+3}(1)(-1)=1$, which is positive. This implies that the steady states $\varphi(z_1),\varphi(z_3)$ are exponentially stable and $\varphi(z_2)$ is unstable. 

We now verify that there exists a set of parameters such that the network has three positive steady states. 
As above, we have 
\[\det(J_{F_T}(\phi))=-\t_1^2\t_3x_1^2x_3+\t_1\t_4\t_5x_1x_2^2+\t_1\t_4\t_5x_2^2x_3+\t_1\t_2\t_3x_1^2+\t_1\t_2\t_4x_1x_2+\t_2\t_4\t_5x_2^2.\]
Taking $\t_i=1$ for $i=1,\ldots ,5$ and $x_1= 5, x_2 = 1, x_3 = 5$, this determinant is equal to -84. By  \cite{FeliuPlos}, the stoichiometric compatibility class containing $\phi(x)$ has more than one positive steady state. The total amounts defining it are $T_1 = 30, T_2 = 1, T_3 = 180$. Using these parameters and solving either $F_T(x)=0$ or $F_{T,1}(\varphi(z))=0$ we prove that there are three positive steady states in his stoichiometric compatibility class, as desired. 

The reduced network is bistable whenever it has three positive steady states, and hence the original network admits bistability in some parameter region.  

\medskip
\textbf{Network (e)}. In this network two substrates \ce{S0} and \ce{P0} are phosphorylated by the same kinase \ce{E}, and dephosphorylated by the same phosphatase \ce{F}.
\begin{align*}
	&\ce{S0 + E <=>[\k_1][\k_2] S0 E ->[\k_3] S1 + E} & 
	&\ce{S1 + F <=>[\k_4][\k_5] S1 F ->[\k_6] S0 + F}\\
	&\ce{P0 + E <=>[\k_7][\k_8] P0 E ->[\k_9] P1 + E} &
	&\ce{P1 + F <=>[\k_{10}][\k_{11}] P1 F ->[\k_{12}] P0 + F}.
\end{align*}
We rename the species as $E = X_1, F = X_2, S_0 = X_3, S_1 = X_4, P_0= X_5 , P_1 = X_6, S_0E = X_7, S_1F = X_8, P_0E = X_9, P_1F = X_{10}$. The polynomial $q_x$ has degree 6, and the first three Hurwitz determinants are positive for every positive steady state. The sign of $H_4$ is unclear, and the analysis of the sign of $H_5$ was interrupted due to memory problems.

The positive steady state variety admits a parametrization $\phi$ by monomials, which in the notation of Section~\ref{sec:comp_challenges}, corresponds to
{\footnotesize
\begin{align*}
	\beta=\left( 1,1,\frac{\k_4\k_6(\k_2+\k_3)}{\k_1\k_3(\k_5+\k_6)}, 1, \frac{(\k_8+\k_9)\k_{10}\k_{12}}{(\k_{11}+\k_{12})\k_7\k_9}, 1, \frac{\k_4\k_6}{\k_3(\k_5+\k_6)} , \frac{\k_{10}\k_{12}}{\k_9(\k_{11}+\k_{12})}, \frac{\k_4}{\k_5+\k_6},\frac{ \k_{10}}{\k_{11}+\k_{12}} \right),
\end{align*}
\[B=\left[ \begin{array}{cccccccccc}
	1 &0& -1& 0& -1& 0& 0& 0& 0& 0\\
	0 &1& 1& 0& 1& 0& 1& 1& 1& 1 \\
	0& 0& 1& 1& 0& 0& 1& 0& 1& 0\\
	0& 0& 0& 0& 1& 1& 0& 1& 0& 1
\end{array} \right] \quad \mbox{\normalsize{ and }}\quad \xi=(x_1,x_2,x_4,x_6). \] }%
Using  tuples, we computed $H_4(\phi)$ and $H_5(\phi)$ and, after studying the signs of their coefficients, we have that $H_4(\phi)$ is positive but the sign of $H_5(\phi)$ is still unclear. 

\smallskip
We proceed to reduce the network by removing all the reverse reactions and the intermediates \ce{S0 E} and \ce{P1 F}. That is, the reactions \ce{S0 + E -> S0 E -> S1 + E} and \ce{P1 + F -> P1 F -> P0 + F} become \ce{S0 + E -> S1 + E} and \ce{P1 + F -> P0 + F} respectively. The reduced network is 
\begin{align*}
	&\ce{S0 + E  ->[\t_1] S1 + E} & 
	&\ce{S1 + F ->[\t_2] S1 F ->[\k_3] S0 + F}\\
	&\ce{P0 + E ->[\t_4] P0 E ->[\k_5] P1 + E}&
	&\ce{P1 + F ->[\t_6] P0 + F}.
\end{align*}
The species are renamed as $E = X_1, F = X_2, S_0 = X_3, S_1 = X_4, P_0 = X_5, P_1 = X_6, S_1F = X_8, P_0E = X_9$. The polynomial $q_x$   has degree 4. The Hurwitz determinants $H_1,H_2$ and $H_3$ are positive for every positive steady state. However, the sign of $H_4$ is unclear even after evaluating at the parametrization $\phi$ below: 
{\small \begin{equation*}
	\phi(x_1,x_2,x_4,x_6)=\left(x_1,x_2,\frac{\t_2x_2x_4}{\t_1x_1},x_4, \frac{\t_6x_2x_6}{\t_4x_1} ,x_6,  \frac{\t_6x_2x_6}{\t_5}, \frac{\t_2x_2x_4}{\t_3} \right).
\end{equation*} }%

The conservation laws of the system are 
\[x_1+x_8=T_1,\quad x_2+x_9=T_2,\quad x_3+x_4+x_9=T_3 \quad \mbox{ and }\quad x_5+x_6+x_8=T_4.\]
Taking $i_1=1,i_2=2,i_3=3,i_4=5$ we define $F_T(x)$ as in Eq.~\eqref{F}. Furthermore, the solutions of $F_{T,\ell}(x)=0$ for $\ell\neq 2$ can be positively parametrized in terms of $z=x_5$ as
{\scriptsize
\begin{align*}
	\varphi (z)=& \left( \frac{T_1\t_5}{\t_4z+\t_5},\frac{-T_1\t_4\t_5z}{\t_6\, b_1(z)}, 
	\frac{(\t_4z+\t_5)T_3\t_2\t_3\t_4z}{b_2(z)}, 
	\frac{-\t_1T_3\t_3\t_6\, b_1(z)}{b_2(z)}, z, 
	\frac{ -b_1(z)}{\t_4z+\t_5},\frac{\t_4T_1z}{\t_4z+\t_5}, 
	\frac{ T_3T_1\t_1\t_2\t_4\t_5z}{b_2(z)} \right)
\end{align*}}%
where
{\small
\begin{align*}
b_1(z) &:= \t_4z^2+ ( (T_1-T_4)\t_4+\t_5)z-T_4\t_5, \\
b_2(z) &:=(-\t_1\t_3\t_4\t_6+\t_2\t_3\t_4^2)z^2+(T_1\t_1\t_2\t_4\t_5-T_1\t_1\t_3\t_4\t_6+T_4\t_1\t_3\t_4\t_6-\t_1\t_3\t_5\t_6+\t_2\t_3\t_4\t_5)z \\ & \qquad +T_4\t_1\t_3\t_5\t_6.
\end{align*}}%
Here $\mathcal{E}=(0,\beta_1)$ where $\beta_1$ is the positive root of $b_1(z)$, such that $b_1(z)<0$ and $b_2(z)>0$ in $\mathcal{E}$.

The positive steady states in the stoichiometric compatibility class are in one to one correspondence with the roots of $F_{T,2}(\varphi(z))$. This is a rational function whose numerator is a polynomial of degree four with positive independent term and   denominator  positive in $\mathcal{E}$. Additionally, $\varphi_5(z)=z$ and $\varphi_5'(z)=1\neq 0$ and taking $i=5,j=2$ we have all the elements in the statement of Theorem~\ref{thm:bistability}. The second hypothesis in the theorem holds, and so does the first as
\[\det(J_F(\varphi(z))_{J,I})=\frac{- T_1\t_5\, b_2(z)}{\t_4z+\t_5}\]
 is negative for  $z\in\mathcal{E}$. The sign of $(F_{T,2}\circ \varphi)'(z_1)$, for $z_1$ the first root of  $(F_{T,2}\circ \varphi)(z_1)$, is positive as the independent term of $(F_{T,2}\circ \varphi)(z)$ is negative. Hence, the sign of Eq.~\eqref{eq:importantsign}
is $(-1)^{5+2+4}(1)(-1)=1$, which is positive. Thus, by  Theorem~\ref{thm:bistability}, we conclude that the steady states $\varphi(z_1),\varphi(z_3)$ are exponentially stable and $\varphi(z_2),\ldots$ are unstable. 

To show that the network admits three positive steady states,  we consider   $\det(J_F(\phi))$:
{\footnotesize
\begin{align*} \det(J_F(\phi))= & \t_1\t_2\t_4\t_6x_1^2x_2^2+\t_1\t_2\t_4\t_6x_1^2x_2x_4+(\t_2\t_5+\t_3\t_6)\t_1 \t_4x_1^2x_2+\t_1\t_2\t_4\t_5x_1^2x_4+\t_1\t_3\t_4\t_5x_1^2  \\
	&  +(\t_1\t_5+\t_3\t_4)\t_2\t_6x_1x_2^2+\t_1\t_2\t_5\t_6x_1x_2x_4+(\t_1\t_6+\t_2\t_4)\t_3\t_5 x_1x_2+\t_1\t_2\t_6^2x_2^3x_6+\\
	& \t_2\t_3\t_5\t_6x_2^2	-\left(\frac{\t_2}{\t_1}-\frac{\t_6}{\t_4}\right)\t_1\t_2\t_4\t_6x_6x_4x_2^2+\t_1\t_3\t_6^2x_6x_2^2+\frac{\t_2\t_3\t_6^2x_2^3x_6}{x_1}.
	\end{align*}}%
Taking $\t_1=\t_3=\t_4=\t_5 =\t_6= 1, \t_2 = 2$ and $x_1 =x_2 = 1, x_4 = 9, x_6 = 9$, the value of the determinant is $-48$. By  \cite{FeliuPlos}, the stoichiometric compatibility class containing $\phi(1,1,9,9)$ has more than one positive steady states and, solving either $F_T(x)=0$ or $F_{T,2}(\varphi(z))=0$ we verify that it has three positive steady states.

 The reduced network is bistable for every set of parameters for which there are three positive steady states, and hence the original network admits bistability in some parameter region.

\medskip
\textbf{Network (f)}. $E$ corresponds to a kinase that exists in two conformations: $E_1$ (relaxed state) and $E_2$ (tensed state) \cite{feng:allosteric}. Each conformation acts as a kinase for a common substrate $S_0$. We denote by  $S_1$ the phosphorylated form of the substrate. We assume that the intermediate kinase-substrate complexes, $E_1S_0$ and $E_2S_0$, also undergo conformational change.
\begin{align*}
	&\ce{E1 + S0 <=>[\k_1][\k_2] E1 S0 ->[\k_3] E1 + S1 }& 
	&\ce{E2 + S0 <=>[\k_4][\k_5] E2 S0 ->[\k_6] E2 + S1}\\
	&\ce{E1 <=>[\k_8][\k_9] E2}& 
	&\ce{E1 S0 <=>[\k_{10}][\k_{11}] E2 S0}& 
	&\ce{S1 ->[\k_7] S0}.
\end{align*}
The species are renamed as $E_1 = X_1, E_2 = X_2, E_1S_0 = X_3, E_2S_0 = X_4, S_0= X_5, S_1 = X_6$. The polynomial $q_x$ associated with this network has degree 4. After computing and evaluating the Hurwitz determinants at a positive parametrization $\phi$ of the positive steady state variety, the sign of $H_3$ and $H_4$ is unclear. Since the hypotheses of Theorem~\ref{thm:bistability} do not hold, we reduce the network by removing all the reverse reactions (with rate constants $\k_2,\k_5,\k_9,\k_{10}$) and the intermediate \ce{E2 S0}. When removing \ce{E2 S0}, the reactions \ce{E2 + S0 -> E2 S0 -> E2 + S1} become \ce{E2 + S0 -> E2 + S1} and the path \ce{E2 + S0 -> E2 S0 -> E1 S0} is collapsed to \ce{E2 + S0 -> E1 S0}. 
The intermediate  \ce{E2 S0} is not of the type considered in the text, where the technical condition on the rate constants for lifting bistability holds. But the condition holds by \cite[Prop. 5.3(iii)]{amir_multi}.

The reduced network is
\begin{align*}
	&\ce{E1 + S0 ->[\t_1] E1 S0 ->[\t_2] E1 + S1 }&
	&\ce{E2 + S0  ->[\t_3] E2 + S1}\\
	&\ce{E2 + S0 ->[\t_6] E1 S0}&
	&\ce{E1 ->[\t_5] E2}&
	&\ce{S1 ->[\t_4] S0}.
\end{align*} 
The species are renamed as $E_1 = X_1, E_2 = X_2, E_1S_0 = X_3, S_0= X_5, S_1 = X_6$. We have $s=3$.  For this reduced network, $q_x$ has degree 3 and the determinants $H_1,H_2$ are positive for every positive steady state. This means that the stability of steady states depends on the sign of $H_3$. We explore the possibility of applying Theorem~\ref{thm:bistability} to ensure bistability. The conservation laws of the network are
$x_1+x_2+x_3=T_1$ and $x_3+x_5+x_6=T_2$.
Taking $i_1=1$ and $i_2=3$ we define $F_T(x)$ as in Eq.~\eqref{F}. The solutions of $F_{T,\ell}(x)=0$ for $\ell\neq 3$ can be parametrized in terms of $z=x_5$ as
{\small  
\[\varphi(z) = \left(\frac{T_1\t_2\t_6z}{b(z)}, \frac{T_1\t_2\t_5}{b(z)},\frac{T_1\t_6z(\t_1z+\t_5)}{b(z)}, z, \frac{T_1\t_2z(\t_1\t_6z+\t_3\t_5+\t_5\t_6)}{\t_4\, b(z)}\right),\] }%
with $b(z)=\t_1\t_6z^2+(\t_2+\t_5)\t_6z+\t_2\t_5$, 
and this parametrization is positive for every $z\in\mathcal{E}=\R_{>0}$. Additionally $\varphi_4(z)=z$ and $\varphi_4'(z)=1\neq 0$. With this parametrization, the positive steady states in one stoichiometric compatibility class are in one to one correspondence with the positive roots of $F_{T,3}(\varphi(x))=0$, that is a rational function whose numerator is a polynomial of degree 3 and has positive denominator: 
{\scriptsize
\begin{equation*}
\begin{split}
	F_{T,3}(\varphi(x))=\frac{1}{(\t_1\t_6z^2 + \t_2\t_6z + \t_5\t_6z + \t_2\t_5)\t_4} & \left[ \t_1\t_4\t_6z^3  \right. + (T_1\t_1\t_2\t_6 + T_1\t_1\t_4\t_6 - T_2\t_1\t_4\t_6 + \t_2\t_4\t_6 + \t_4\t_5\t_6)z^2 + \\ & \hspace{-1.5cm} (T_1\t_2\t_3\t_5 + T_1\t_2\t_5\t_6 + T_1\t_4\t_5\t_6 - T_2\t_2\t_4\t_6 - T_2\t_4\t_5\t_6 + \t_2\t_4\t_5)z - \left. T_2\t_2\t_4\t_5 \right] .
\end{split}
\end{equation*} }%
Here $i=4,j=3$. The second hypothesis of Theorem~\ref{thm:bistability} holds. For the first hypothesis, we have
\[\det(J_F(\varphi(z)_{J,I})= \t_4\,  b(z),\]
which is positive for every $z\in\mathcal{E}$.  The sign of $(F_{T,3}\circ \varphi)'(z_1)$, for $z_1$ the first positive root of $F_{T,3}(\varphi(z))$, is positive as the independent term of $(F_{T,3}\circ \varphi)(z)$ is negative. 
 Furthermore, the sign of Eq.~\eqref{eq:importantsign}
 is $(-1)^{3+3+4}(1)(1)=1$ positive. By  Theorem~\ref{thm:bistability},  the steady states $\varphi(z_1),\varphi(z_3),\dots$ are exponentially stable and $\varphi(z_2),\dots$ are unstable. 

To verify the existence of three positive steady states, we consider  $\det(J_F(\phi))$ for $\phi$ as follows:
{\small 
\begin{align*}
	\phi(x_2,x_5)& =\left( \frac{ \t_6x_2x_5}{\t_5}, x_2, \frac{\t_6x_2x_5(\t_1x_5 + \t_5)}{\t_2\t_5}, x_5, \frac{x_2x_5(\t_1\t_6x_5 + \t_3\t_5 + \t_5\t_6)}{\t_4\t_5}  \right).
\\
\det(J_F(\phi)) &=\t_1 (\t_2\t_6-\t_3\t_5+\t_4\t_6)\t_6 x_2x_5^2 +2(\t_2+\t_4)\t_1\t_5\t_6x_2x_5+(\t_2\t_3+\t_2\t_6+\t_4\t_6)  \t_5^2 x_2  \\
	&\qquad   + \t_1\t_4\t_5\t_6 x_5^2+(\t_2+\t_5)\t_4\t_5\t_6x_5+\t_2\t_4\t_5^2.
\end{align*}}%
Taking $\t_i=1$ for $i=1,\ldots ,6$ and $x_2 = 2, x_5 = 6$, $\det(J_F(\phi)) =-35$. By \cite{FeliuPlos}, the stoichiometric compatibility class containing $\phi(x)$ has more than one positive steady state. This class is defined by the total amounts $T_1 = 98, T_2= 222$ and by solving $F_T(x)=0$ or $F_{T,3}(\varphi(z))=0$, we verify that there are three positive steady states. 

Hence the reduced network is bistable for every set of parameters for which there are three positive steady states, and the original network admits bistability in some   parameter region.

\end{document}